\documentclass{article}
\usepackage[utf8]{inputenc}
\usepackage{url} 

\usepackage{graphicx}
\usepackage{amsmath}
\usepackage{amsthm}
\usepackage[english]{babel}
\usepackage{caption}
\usepackage{subcaption}
\usepackage[table]{xcolor}
\usepackage[ruled,vlined]{algorithm2e}
\usepackage{tablefootnote}
\usepackage{amssymb}
\usepackage{bbm}

\newcommand{\hf}{{\textstyle{{\frac{1}{2}}}}}

\newcommand{\bx}{\mbox{\boldmath$x$}}

\newcommand{\btheta}{\mbox{\boldmath$\theta$}}
\newcommand{\bnu}{\mbox{\boldmath$\nu$}}
\newcommand{\ba}{\mbox{\boldmath$a$}}
\newcommand{\bb}{\mbox{\boldmath$b$}}
\newcommand{\bu}{\mbox{\boldmath$u$}}
\newcommand{\bv}{\mbox{\boldmath$v$}}
\newcommand{\bz}{\mbox{\boldmath$z$}}
\newcommand{\bone}{\mbox{\boldmath$1$}}

\newcommand{\iu}{\boldsymbol{\mathrm{i}}}
\newcommand{\bpsi}{\mbox{\boldmath$\psi$}}
\DeclareUnicodeCharacter{2212}{-}

\DeclareMathOperator{\R}{\mathbb{R}}
\DeclareMathOperator{\C}{\mathbb{C}}
\DeclareMathOperator{\I}{\mathbbm{1}}

\newtheorem{theorem}{\bf Theorem}[section]
\newtheorem{proposition}{\bf Proposition}[section]

\newtheorem{remark}{\bf Remark}[section]
\newtheorem{assumption}{\bf Assumption}[section]
\definecolor{brilliantrose}{rgb}{1.0, 0.33, 0.64}
\definecolor{myviolet}{rgb}{0.21, 0.0, 0.85}
\definecolor{amethyst}{rgb}{0.6, 0.4, 0.8}
\definecolor{carrotorange}{rgb}{0.93, 0.57, 0.13}

\author{
  Xue Gong
  \thanks{%
           School of Mathematics,
           University of Edinburgh,
           Edinburgh, EH9 3FD, UK; The Maxwell Institute for Mathematical                Sciences, Edinburgh, EH8 9BT, UK
           (\texttt{X.Gong-8@sms.ed.ac.uk})
           }
           \and
        Desmond J. Higham%
            \thanks{%
           School of Mathematics,
           University of Edinburgh,
           Edinburgh, EH9 3FD, UK
           (\texttt{d.j.higham@ed.ac.uk})
           }
           \and
           Konstantinos Zygalakis
            \thanks{%
           School of Mathematics,
           University of Edinburgh,
           Edinburgh, EH9 3FD, UK
           (\texttt{K.Zygalakis@ed.ac.uk})
           }
        }

\date{}
\title{Generative Hypergraph Models and Spectral Embedding} 

\begin{document}
\maketitle
\begin{abstract}
Many complex systems involve interactions between more than two agents. Hypergraphs capture these higher-order interactions through hyperedges that may link more than two nodes. We consider the problem of embedding a hypergraph into low-dimensional Euclidean space so that most interactions are short-range. This embedding is relevant to many follow-on tasks, such as node reordering, clustering, and visualization. We focus on two spectral embedding algorithms customized to hypergraphs which recover linear and periodic structures respectively. In the periodic case, nodes are positioned on the unit circle. We show that the two spectral hypergraph embedding algorithms are associated with a new class of generative hypergraph models. These models generate hyperedges according to node positions in the embedded space and encourage short-range connections. They allow us to quantify the relative presence of periodic and linear structures in the data through maximum likelihood. They also improve the interpretability of node embedding and provide a metric for hyperedge prediction. We demonstrate the hypergraph embedding and follow-on tasks---including quantifying relative strength of structures, clustering and hyperedge prediction---on synthetic and real-world hypergraphs. We find that the hypergraph approach can outperform clustering algorithms that use only dyadic edges. We also compare several triadic edge prediction methods on high school and primary school contact hypergraphs where our algorithm improves upon benchmark methods when the amount of training data is limited.
\end{abstract}

\section{Introduction}
\label{sec:intro}

A typical graph-based data set captures pairwise interactions between 
nodes. There is growing interest in understanding higher-order, group-level, interactions, with different paradigms being proposed \cite{bianconi2021higher, lambiotte2019networks}.
In this work, we represent such interactions with a
hypergraph formulation; here each hyperedge
involves two or more nodes. This framework is discussed in
\cite{benson2016higher,benson2021higher,torres2020why}
 and has found application in real-world problems such as epidemic spread modelling \cite{higham2021epidemics}, image classification \cite{Yu2012image}, and the study of biological networks \cite{Ramadan2004yeast}.

A fundamental learning task on graph-based data is to embed nodes into a low-dimensional Euclidean space \cite{Scholkopf2007learn, Rossi2020}. The learned embedding could be used in follow-on tasks  
such as clustering, classification, and structure recovery. There are various types of learning algorithms for graphs; some design and analyze Laplacian matrices related to the graph \cite{LuxburgUlrike2007Atos}, some solve maximum likelihood problems associated with random graph models \cite{GrindrodPeter2002Rrga, peixoto2022ordered}, and others involve more complex machine learning frameworks \cite{hamiltonGRL, hamilton2017representation}. 

In this work, we build on the use of spectral  methods which derive node embeddings from eigenvectors of a Laplacian matrix \cite{ChungFanR.K.1997Sgt}. 
Such spectral algorithms are popular, since they can be implemented efficiently on large sparse graphs
and they are backed up by 
accompanying consistency theory \cite{dKH22}.
Two main approaches have also recently been proposed for spectral clustering on hypergraphs. 
One approach is to employ higher-order Laplacian tensors \cite{galuppi2021spectral}. Tensors in general contain richer information, however, their use  
can require considerably more computational expense than matrix algorithms, and the results 
can be difficult to visualize and interpret.
A second approach is to ``flatten''
the higher order information into a representative 
node-level matrix. Some matrix-based approaches analyze the vertex-edge incidence matrix associated with a random walk interpretation \cite{Scholkopf2007learn}, other frameworks utilise motif-based Laplacian matrices that could be generalized to various motifs and time steps \cite{Rossi2020, Lucas2020}. 
The methodology that we develop here fits into this second category by building a node-based matrix, using an intermediate step that looks over all 
hyperedge dimensions in order 
to incorporate higher order information.

A second aspect of our work is the connection between spectral methods and random models.
 Many graph embedding \cite{hoff2002latent}, re-ordering \cite{GrindrodPeter2002Rrga}, clustering \cite{Chodrow21}, and structure recovery \cite{cimini2021reconstructing, tudisco2022core} techniques solve maximum-likelihood problems on graphs assuming specific generative models. Besides their application in these inverse problems, random graph models are useful inference tools for quantifying structure, 
  predicting new or missing links, and improving the interpretability of learning algorithms by relating node embeddings to edge probabilities
  \cite{GrindrodPeter2010Pr, GHZ21}. 
 
 Many spectral algorithms are naturally related to optimization problems. This is the case when the Laplacian matrix is Hermitian, so that its eigenvectors are critical points of a quadratic form \cite{LutkepohlHelmut1996Hom}. For example, spectral embedding for undirected graphs using the standard combinatorial Laplacian is related to minimizing the unnormalized cut \cite{LuxburgUlrike2007Atos,spec_hkk}. 
 Furthermore, such optimization formulations may lead to  interesting random graph interpretations of spectral algorithms. When the quadratic form can be expressed as the log-likelihood of the graph under a suitable model, the optimization problem may be restated as a node reordering. Such connections have been investigated for undirected graphs \cite{HighamDesmondJ2003Uswn} and directed graphs \cite{GHZ21}, 
 and here we extend these ideas to the hypergraph 
 setting.
 In particular, we associate customized spectral embedding algorithms with generative models that belong to a new class of range-dependent random hypergraphs that encourages short-range connections between nodes, generalizing existing graph models  \cite{GrindrodPeter2002Rrga, higham2005spectral}. These range-dependent random hypergraphs 
 offer flexibility that is not available in
 stochastic block models \cite{Chodrow21}
 which require block sizes to be pre-specified or inferred.

The rest of the paper is structured as follows. Our notation is introduced in Section~\ref{sec:notation}.
In Sections~\ref{sec:linear_embedding} and \ref{sec:periodic_embedding} we define the 
linear and periodic hypergraph embedding algorithms and derive  associated optimization problems. 
We propose random models associated with the hypergraph embedding algorithms in 
Section~\ref{sec:model}, which leads to a model comparison workflow that quantifies the relative strength of linear versus periodic structures. Numerical studies on synthetic and real-world hypergraphs using the proposed models  are presented in Section~\ref{sec:exp}.

The main contributions of this work are as follows.
\begin{itemize}
\item We propose new range-dependent generative models for hypergraphs that generate linear and periodic cluster patterns. 
\item We establish their connection with linear and periodic spectral embedding algorithms. 
\item We demonstrate on synthetic and real data that, after tuning model parameters to the data, these models can quantify the relative 
strength of linear and periodic structures.
\item We perform prediction of triadic hyperedges
(triangles) using the proposed linear model and show that it outperforms the existing average-score based method \cite{benson2018simplicial} on synthetic hypergraphs, and also on high school and primary school contact data when the amount of training data is limited.
\end{itemize}

\section{Notation}
\label{sec:notation}
We consider \textit{undirected, unweighted} hypergraphs $G = (V, E)$ on the vertex set $V$ containing $n$ nodes and the hyperedge set $E$. 
We let $R \in \mathcal{R}$ be an unordered set of nodes, where $\mathcal{R}$ denotes the collection of all such sets. 
We use $\lvert R \rvert$ to denote the number of  nodes in tuple $R$, that is, its cardinality, and we assume $ 2 \leq \lvert R \rvert \leq T$ for all $R \in E$.

Let $A_R$ indicate the presence of a hyperedge, so that $A_R = 1$ if $R \in E$ and $A_R = 0$ otherwise. We define the $t$-th order 
$n$ by $n$ adjacency matrix $W^{[t]}$ such that $W^{[t]}_{ij}$ counts the number of hyperedges with cardinality $t$ that contain distinctive nodes $i$ and $j$; hence, $W^{[t]}_{ij} = \sum_{R \in \mathcal{R}: |R| =t} A_R  \I{( i \in R)}\I{( j \in R )}$ if $i \neq j$, and  $W^{[t]}_{ij} = 0$ otherwise, where $\I$ is the indicator function. 
Similarly, we define the corresponding $t$-th order diagonal degree matrix $D^{[t]}$ such that $D^{[t]}_{ii} = \sum_{j\in V} W^{[t]}_{ij}$, and
the $t$-th order Laplacian matrix $L^{[t]} = D^{[t]} - W^{[t]}$.
We use $\iu$ to denote $\sqrt{-1}$ and $\bone$ to denote the vector in
$\R^n$ with all entries equal to one.
We let $\ba'$ represent the transpose of a 
real-valued vector $\ba$ and let $\bb^H$ denote the conjugate transpose of a complex-valued 
vector $\bb$. 
We use $\mathcal{P}$ to denote the set of all permutation vectors, that is, all vectors in $\R^n$ that contain each of the integers $1,2,\ldots,n$. We will focus on one-dimensional embedding.
We let $x_i \in \R$ be the location to which node $i$ is assigned, and $\bx = [x_1, x_2, ..., x_n]' \in \R^n$ .

\section{Linear hypergraph embedding}
\label{sec:linear_embedding}
Given a hypergraph, suppose we wish to find node embeddings $\bx \in \R^n $ such that 
hyperedges tend to contain nodes that are a small distance apart.
To formalize this idea, we can define a linear incoherence function $I_{\text{lin}}(\bx, R)$ that sums up the squared Euclidean distance between all nodes pairs in tuple $R$: 
\begin{equation}
I_{\text{lin}}(\bx, R) = \sum_{i,j\in R}(x_i -x_j)^2.
\label{eq:linear_incoherence}
\end{equation}
We may then define the total linear incoherence of the hypergraph, $\eta_{\text{lin}}(G,\bx)$, by aggregating the linear incoherence over all node tuples. Furthermore, we may wish to tune the weights of hyperedges of different cardinalities through a coefficient $c_{\lvert R \rvert} \ge 0$ for node tuple $R$; that is,  
\begin{equation}
    \eta_{\text{lin}}(G,\bx) = \sum_{R \in \mathcal{R}} c_{\lvert R \rvert} A_R I_{\text{lin}}(\bx, R).
    \label{eq:eta_linear}
\end{equation}
One justification for these tuning parameters  $c_{\lvert R \rvert}$ is that they allow us to avoid 
the case where high-cardinality hyperedges dominate the expression.
For example, we could choose $c_{t} = \frac{1}{t(t-1)} $ to balance the contributions from hyperedges with different sizes. A suitable choice of $c_{t}$ may also depend on the relative importance of hyperedges in the application.

In Proposition~\ref{prop:quad} 
we show that the total linear incoherence
maybe be written 
as a quadratic form involving the 
 hypergraph Laplacian matrix
\begin{equation}
    L = \sum_{t=2}^T c_t L^{[t]}. 
    \label{def:hypergraph_Laplacian}
\end{equation}

\begin{proposition}\label{prop:quad}
For any $\bx\in\R^n$ with $L$ defined in (\ref{def:hypergraph_Laplacian}), and 
$\eta_{\text{lin}}(G,\bx)$ defined in (\ref{eq:eta_linear}) we have 
\begin{equation}
\bx'L\bx = \hf\eta_{\text{lin}}(G,\bx).
\label{eq:hypergraph_objective}
\end{equation}
\end{proposition}

\begin{proof}
It is straightforward to show that  $\bx'L^{[t]}\bx = \sum_{i,j\in V} x_i (D^{[t]}_{ij} - W^{[t]}_{ij}) x_j= \hf\sum_{i,j\in V} W^{[t]}_{ij}(x_i - x_j)^2$. Therefore,
\begin{align*}
    \bx' L \bx 
    &= \sum_{t=2}^T \bx' c_t L^{[t]} \bx\\
    &= \hf\sum_{t=2}^T  c_t \sum_{i,j\in V} W^{[t]}_{ij}(x_i - x_j)^2 \\
    &= \hf\sum_{t=2}^T  c_t \sum_{i,j\in V} \sum_{R \in \mathcal{R}: |R| =t} A_R  \I{( i \in R )}\I{( j \in R )} (x_i - x_j)^2\\
    &= \hf\sum_{R \in \mathcal{R}}c_{\lvert R \rvert} A_R  \sum_{i,j\in R}  (x_i - x_j)^2 = \hf\eta_{\text{lin}}(G,\bx).
\end{align*}
\end{proof}
We note that each Laplacian $L^{[t]}$
is symmetric and positive semi-definite with smallest eigenvalue $0$.

\begin{assumption}
We assume that the 
unweighted, undirected graph associated with 
the binarized version of $L$ is connected.
It then follows that
$L$ has a single eigenvalue equal to $0$ with  
all other eigenvalues positive. We further assume that 
there is a unique smallest positive eigenvalue,
$\lambda_2$.
(The eigenvector $\bv^{[2]}$ corresponding to $\lambda_2$
is a generalization of the classic Fiedler vector.)
\label{assumption:connected}
\end{assumption}

In minimizing 
the 
total linear incoherence
(\ref{eq:eta_linear}) we must avoid
the trivial cases where 
(a) all nodes are located arbitrarily close to the origin and (b) all nodes are assigned to 
the same location. Hence it is natural to impose the constraints 
$\| \bx \|_2 =1$ and $\bx' \bone = 1$.
It then follows from the Rayleigh-Ritz Theorem
\cite[Theorem~4.2.2]{HJ85}
that the quadratic form in 
Proposition~\ref{prop:quad} is solved by 
$\bx = \bv^{[2]}$. This leads us to Algorithm~\ref{algo:lin_spec_embed}
below, which could also be considered as a special case of the algorithm in \cite{galuppi2021spectral} where the motifs considered are hyperedges.

\begin{algorithm}[H]
\SetAlgoLined
\KwResult{Node embedding $\bx \in \R^{n}$}
\textbf{Input hyperedge adjacency matrices} $W^{[2]},W^{[3]}, ...,W^{[T]}$\;
\textbf{Construct diagonal degree matrices} $D^{[t]}_{ii} = \sum_{j\in V} W^{[t]}_{ij}$\;
\textbf{Construct $t$-th order Laplacians} $L^{[t]} = D^{[t]} - W^{[t]}$\;
\textbf{Construct hypergraph Laplacian} $L = \sum_{t=2}^T c_t L^{[t]}$\;
\textbf{Compute second smallest eigenvalue
$\lambda_2$ and corresponding eigenvector $\bv^{[2]}$}\;
\textbf{Embed nodes using} $\bx =
\bv^{[2]}$
\caption{Linear Hypergraph Embedding Algorithm}
\label{algo:lin_spec_embed}
\end{algorithm}

\begin{remark}
Algorithm \ref{algo:lin_spec_embed} could be extended to higher dimensional embeddings where node $i$ is assigned to $\bx^{[i]} \in \R^{d}$ for  $d >1$. In this case we could generalize (\ref{eq:linear_incoherence}) to  
\begin{equation}
I_{\text{lin}}(\bx, R) = \sum_{i,j\in R} \lVert \bx^{[i]} -\bx^{[j]}\rVert_2^2.
\label{eq:higher_order_lin_incoherence}
\end{equation}
If we require the coordinate directions to be orthogonal, then 
the embedding is found via the eigenvectors corresponding to the $d$-smallest non-zero eigenvalues; see
\cite{spec_hkk} for details in the graph case.
\end{remark}

\section{Periodic hypergraph embedding}
\label{sec:periodic_embedding}
In this section, we look at the periodic analogue of linear hypergraph embedding. Here nodes are embedded into the unit circle rather than along the real line.
Such a periodic
structure formed the basis
of the 
classic ``small world'' model of Watts and Strogatz \cite{WattsD.J1998Cdos}.
Results in \cite{GrindrodPeter2010Pr} showed that 
certain real networks are better represented
via this type of ``wrap-around'' notion of distance.
Hence, it is of interest to develop concepts that apply to the hypergraph case.

We may position nodes on the unit circle by mapping them to phase angles $\btheta = \{ \theta_i \}_{i=1}^{n} \in [0, 2\pi)$. We may then use a periodic incoherence function to quantify the distance between node pairs in the tuple $R$:
\begin{equation}
I_{\text{per}}(\btheta, R) = \sum_{i,j\in R}|e^{\iu\theta_i}-e^{\iu\theta_j}|^2.
\label{eq:periodic_incoherence}
\end{equation}
Then the total periodic incoherence of the hypergraph  becomes
\begin{equation}
\eta_{\text{per}}(G,\btheta) = \sum_{R \in \mathcal{R}} c_{\lvert R \rvert} A_R  I_{\text{per}}(\btheta, R). 
\label{eq:graph_incoherence_periodic}
\end{equation}
In Proposition~\ref{prop:per} below, we relate
the total periodic incoherence to a quadratic form
involving the hypergraph 
 Laplacian matrix (\ref{def:hypergraph_Laplacian}).
\begin{proposition}\label{prop:per}
Let $\bpsi \in \C^n$ be such that $\psi_j = e^{\iu\theta_j}$. Then
\begin{equation}
\bpsi^{H}L\bpsi =\hf\eta_{\text{per}}(G,\btheta).
\label{eq:qper}
\end{equation} 
\end{proposition}

\begin{proof}
We have 
\begin{align*}
    \bpsi^{H}L^{[t]}\bpsi &= \bpsi^{H}D^{[t]}\bpsi -\bpsi^{H}W^{[t]}\bpsi\\
    &= \sum_{i\in V} e^{-\iu\theta_i} (\sum_{j\in V} W^{[t]}_{ij}) e^{\iu\theta_i} - \sum_{i, j\in V} e^{-\iu\theta_i}  W^{[t]}_{ij} e^{\iu\theta_j}\\
    &=\sum_{i, j\in V} W^{[t]}_{ij} (1-e^{-\iu\theta_i}   e^{\iu\theta_j})\\
    &=\hf\sum_{i, j\in V} W^{[t]}_{ij} (2-e^{-\iu\theta_i}e^{\iu\theta_j}-e^{-\iu\theta_j}e^{\iu\theta_i})\\
    &=\hf\sum_{i, j\in V} W^{[t]}_{ij} (e^{\iu\theta_i}-e^{\iu\theta_j})(e^{-\iu\theta_i}-e^{-\iu\theta_j})\\
    &=\hf  \sum_{i,j\in V} W^{[t]}_{ij} |e^{\iu\theta_i}-e^{\iu\theta_j}|^2.
\end{align*}
Then the proof may be completed in a similar way to the 
proof of Proposition~\ref{prop:quad}.
\end{proof}

Appealing again to the Rayleigh–Ritz theorem
\cite[Theorem~4.2.2]{HJ85},
the quadratic form in
(\ref{eq:qper}) 
is minimized over all
$\bpsi \in \C^n$
with 
$\| \bpsi \|_2 = 1$ and 
$\bpsi^H \bone = 1$ by taking 
$\bpsi = \bv^{[2]}$.
However, this real-valued eigenvector cannot be
proportional to a vector with components of the 
form $e^{\iu\theta_j}$.
Hence, following the approach in
\cite{GrindrodPeter2010Pr}
we will use the
heuristic of 
setting 
\begin{equation}
\theta_i = \text{angle}(v^{[2]}_i+ \iu v^{[3]}_i) \in [-\pi, \pi],
\label{eq:inverse_tangent}
\end{equation}
defined as $v^{[2]}_i+ \iu v^{[3]}_i = \lvert v^{[2]}_i+ \iu v^{[3]}_i \rvert \cdot e^{i\theta_i}$, where $ \bv^{[3]}$ is an eigenvector corresponding to the next-smallest
eigenvalue of $L$. Such a heuristic assumption converts two real eigenvectors into a complex vector, which gives an approximate solution to the minimization problem. The resulting workflow is summarized in 
Algorithm~\ref{algo:per_spec_embed}.

We also note that for simple unweighted, undirected graphs, finding $\btheta$ that minimizes $\eta_{\text{per}}(G,\btheta)$ is equivalent to the formulation proposed in \cite{GrindrodPeter2010Pr}. This may be shown by letting $u_i = \cos \theta_i$ and $z_i = \sin \theta_i$ and expanding
(\ref{eq:graph_incoherence_periodic}) as
\[
 \sum_{i \in V} \sum_{j \in V} W^{[2]}_{ij} \left( (u_i - u_j)^2 + (z_i - z_j)^2 \right),
 \]
 which simplifies to
 \begin{equation}
 2 \bu^T (D^{[2]} -W^{[2]}) \bu + 2 \bz ^T (D^{[2]} -W^{[2]}) \bz.
 \label{eq:Lapeq}
 \end{equation}
 This is essentially equation (3.1) in \cite{GrindrodPeter2010Pr}, derived from a slightly different perspective.

We then arrive at Algorithm~\ref{algo:per_spec_embed} below.

\begin{algorithm}[H]
\SetAlgoLined
\KwResult{Node embedding $\btheta = \{ \theta_i \}_{i=1}^{n} \in [0, 2\pi)$}
\textbf{Input adjacency matrices} $W^{[2]},W^{[3]}, ...,W^{[T]}$\;
\textbf{Construct diagonal degree matrices} $D^{[t]}_{ii} = \sum_{j\in V} W^{[t]}_{ij}$\;
\textbf{Construct $t$-th order Laplacians} $L^{[t]} = D^{[t]} - W^{[t]}$\;
\textbf{Construct hypergraph Laplacian} $L = \sum_{t=2}^T c_t L^{[t]}$\;
\textbf{Compute second and third smallest eigenvalues $\lambda_2$ and $\lambda_3$ and corresponding eigenvectors 
$ \bv^{[2]}$
and
$ \bv^{[3]}$}\;
\textbf{Calculate phase angles} $\theta_i=\text{angle}(v^{[2]}_i+ \iu v^{[3]}_i)$\;
\textbf{Embed nodes using} $\btheta$ 
\caption{Periodic Hypergraph Embedding Algorithm}
\label{algo:per_spec_embed}
\end{algorithm}

\section{Generative hypergraph models}
\label{sec:model}

We now discuss a connection between the minimization of total incoherence and generative models.  Let us consider finding a node embedding $\bx \in \R^n$ that minimizes a generic total graph incoherence expression 
\begin{equation}
\eta(G) = \sum_{R \in \mathcal{R}} c_{\lvert R \rvert} A_R I(\bx, R),
\label{eq:geninc}
\end{equation}
for a non-negative incoherence function $I(\bx, R)$.
We consider the case where the $x_i \in \R$ must take 
distinct values from a discrete set $\{ \nu_i \}_{i=1}^n$, where $\nu_i \in \R$; that is,
we must have 
$x_i = \nu_{p_i}$, where $p \in {\mathcal{P}}$ is a permutation vector.
In the linear case, this set may be the integers from $1$ to $n$ and in the periodic case
this set may be equally spaced angles
in $[0,2\pi)$.

Now consider a random hypergraph model where each hyperedge involving node tuple $R\in \mathcal{R}$ is generated independently with probability
\begin{equation}
\textbf{P}(A_R = 1) = f_{R}(\bx, R),
\label{def:unweighted_model_form}
\end{equation}
for a function $f_{R}$ that takes values between 0 and 1. We have the following connection. 
\begin{theorem} 
Suppose $\bx \in \R^n$ is constrained to take values
from a discrete set such that $x_i = \nu_{p_i}$, where $p \in {\mathcal{P}}$ is a permutation vector. Then minimizing the total incoherence (\ref{eq:geninc}) over all such $\bx$ is equivalent to maximizing over all such $\bx$ the likelihood that the hypergraph is generated by a model of the form (\ref{def:unweighted_model_form}), where
\begin{equation}
f_{R}(\bx, R) = \frac{1}{1+e^{\gamma c_{\lvert R \rvert} I(\bx, R)}}
\label{eq:unweighted_model}
\end{equation}
for any positive $\gamma$.
\label{thm:like}
\end{theorem}

\begin{proof}
Using (\ref{def:unweighted_model_form}),
the likelihood of the whole hypergraph is
\begin{align*}
L(G) &= \prod_{R \in \mathcal{R}:A_R = 1} f_{R}(\bx, R)\prod_{R \in \mathcal{R}:A_R = 0} \left(1-f_{R}(\bx, R)\right) \\
&=\prod_{R \in \mathcal{R}:A_R = 1} \frac{f_{R}(\bx, R)}{1-f_{R}(\bx, R)}\prod_{R \in \mathcal{R}} \left(1-f_{R}(\bx, R)\right), 
\end{align*}
which leads to the log-likelihood
\begin{align}
\ln (L(G)) &=\sum_{R \in \mathcal{R}:A_R = 1} \ln\left(\frac{f_{R}(\bx, R)}{1-f_{R}(\bx, R)}\right)+\sum_{R \in \mathcal{R}} \ln \left(\left(1-f_{R}(\bx, R)\right) \right).
\label{eq:graph_logP}
\end{align}
The second term on the right-hand side, which is the probability of the null hypergraph, is independent of the the permutation. Hence,
with 
(\ref{eq:unweighted_model}), 
maximizing the log-likelihood of the hypergraph is equivalent to minimizing
\begin{align}
\sum_{R \in \mathcal{R}:A_R = 1} \ln\left(\frac{1-f_{R}(\bx, R)}{f_{R}(\bx, R)}\right) =  \sum_{R \in \mathcal{R}} c_{\lvert R \rvert} A_R \gamma I(\bx, R) = \gamma \, \eta(G).
\end{align}
\end{proof}

\begin{remark}
Theorem \ref{thm:like} could be extended to the case where node $i$ is assigned to $\bx^{[i]} \in \R^{d}$ for  $d >1$, and a higher-dimensional incoherence function in (\ref{eq:higher_order_lin_incoherence}) is considered. In this scenario, we constrain $\bx^{[i]} \in \R^{d}$ to take values from a discrete set $\{ \bnu^{[i]} \}_{i=1}^n$ where $\bnu^{[i]} \in \R^d$, such that $\bx^{[i]} = \bnu^{[p_i]}$ for a permutation vector $p \in {\mathcal{P}}$. Then we could follow the same arguments as in Theorem \ref{thm:like} to derive a model described by (\ref{def:unweighted_model_form}) and (\ref{eq:unweighted_model}), where $\bx \in \R^{n\times d}$ and $I(\bx, R) = \sum_{i,j\in R} \lVert \bx^{[i]} -\bx^{[j]}\rVert_2^2$. 
\end{remark}

For a hypergraph generated by model (\ref{eq:unweighted_model}) the number of hyperedges connecting the node tuple $R$ follows a Bernoulli distribution with probability $1/1+e^{\gamma c_{\lvert R \rvert} I(\bx, R)}$. The log-odds of the hyperedge decay linearly with the incoherence of the node tuple since

\begin{equation*}
\ln(f_{R}(\bx, R) /(1-f_{R}(\bx, R) )) = -\gamma c_{\lvert R \rvert} I(\bx, R),
\end{equation*}
where the factor $\gamma c_{\lvert R \rvert}$ determines the decay rate. The probability of a hyperedge is highest when all nodes overlap, i.e., $I(\bx, R)=0$, which gives a $1/2$ likelihood. 
If we generate hyperedges in repeated trials for the node tuple $R$, the variance of the number of hyperedges is $e^{c_{\lvert R \rvert}\gamma I(\bx, R)}/(1+e^{c_{\lvert R \rvert}\gamma I(\bx, R)})^2$. When $I(\bx, R)=0$, the largest variance of $1/4$ is achieved. The expected total number of hyperedges of the whole hypergraph $G$ can be expressed as  
\begin{equation*}
\sum_{R\in \mathcal{R}} f_R(\bx, R) = \sum_{R\in \mathcal{R}} \frac{1}{1+e^{\gamma c_{\lvert R \rvert} I(\bx, R)}}.
\end{equation*} 

We note that Theorem~\ref{thm:like} introduces
the extra scaling parameter $\gamma$.
This parameter plays no direct role in Algorithms~\ref{algo:lin_spec_embed} and \ref{algo:per_spec_embed}. However, a value for $\gamma$ is needed if we wish to compare
the likelihoods of the two models having inferred
the embeddings.
In principle, we may fit the parameter $\gamma$ to a given hypergraph by matching the observed number of hyperedges with their expectation. However, from a computational point of view, this is rather challenging in general, since the computational complexity of the expectation calculation is $\mathcal{O}(n^T)$ when the maximum cardinality of a considered hyperedge is $T$. Hence, 
given an embedding, in practice we prefer to pick $\gamma$ by maximizing the likelihood, as described in the following subsection.

\subsection*{Model comparison}
\label{subsec:comp}
Under the assumption that a given hypergraph arose from a mechanism that favours connections between ``nearby'' nodes (in some latent, unobservable configuration), it is of interest to know whether a linear or periodic distance provides a better description.We may address this question using a model comparison approach. As in Sections~\ref{sec:linear_embedding} and \ref{sec:periodic_embedding}, we consider one-dimensional embeddings, such that both the linear and periodic version have $n+1$ parameters given the Laplacian coefficients: $n$ node embeddings and a decay parameter $\gamma$. The node embeddings will be estimated from Algorithms~\ref{algo:lin_spec_embed} and \ref{algo:per_spec_embed}. For any choice of $\gamma$, we may then calculate the corresponding likelihood for each type of hypergraph, given the embedding. We may then compare the models by reporting plots of likelihood versus $\gamma$ or by reporting the maximum likelihood over all $\gamma$. We note that Theorem \ref{thm:like} states that node embeddings that minimize the incoherence also maximize the graph likelihood under the given discrete constraints. We note that Algorithms~\ref{algo:lin_spec_embed} and \ref{algo:per_spec_embed} minimize linear and periodic incoherence after relaxing the discrete constraints in Theorem \ref{thm:like} for computational feasibility. Such heuristics are often used in discrete programming. Therefore instead of the exact maximum likelihood, we get an estimated maximum likelihood.
An overall workflow is shown below in Algorithm~\ref{algo:model_comparison}.\\
\begin{algorithm}[H]
\SetAlgoLined
\KwResult{Comparison of possible graph structures}
\textbf{Input hypergraph} $G$\;
\textbf{Compute linear embedding using Algorithm \ref{algo:lin_spec_embed}}\;
\textbf{Compute periodic embedding using Algorithm \ref{algo:per_spec_embed}}\;
\textbf{Calculate maximum likelihood of linear model over $\gamma >0$} using (\ref{eq:graph_logP}), (\ref{eq:unweighted_model}), and (\ref{eq:linear_incoherence}) \;
\textbf{Calculate maximum likelihood of periodic model over $\gamma >0$} using (\ref{eq:graph_logP}), (\ref{eq:unweighted_model}), and (\ref{eq:periodic_incoherence}) \;\textbf{Compare likelihoods or report maxima} 
\caption{Model Comparison}
\label{algo:model_comparison}
\end{algorithm}

\section{Experiments}
\label{sec:exp}

\subsection*{Model Comparison}
\subsubsection*{Synthetic Hypergraphs}
In this section we test the performance on 
Algorithms~\ref{algo:lin_spec_embed}, 
\ref{algo:per_spec_embed} and \ref{algo:model_comparison} in a controlled setting.
To do this, we generate hypergraphs with 
either linear or periodic clustered structure using
the proposed random model. For simplicity, we only consider dyadic and triadic edges, although the experiments could be extended to include higher-order hyperedges.

\paragraph{Linear hypergraph with clustered nodes}
\label{subsec:synthetic_linear}
We first generate hypergraphs with $K$ planted clusters  $C_1, C_2, ..., C_K$ of size $m$, and  $n = m K$ nodes. 
We embed the nodes
using 
 $x_{i} = \frac{2(l-1)}{K} + \sigma$ if $i\in C_l$,  where $\sigma \sim 
\mathrm{unif}(-a,a)$ is an additive uniform noise. Hyperedges are then drawn randomly according to  model (\ref{eq:unweighted_model}) with the linear incoherence (\ref{eq:linear_incoherence}). 

We note that, in practice, the embedding algorithms
must choose values $c_2$ and $c_3$ in order to 
form the hypergraph Laplacian,
and the model comparison algorithm must 
choose a value for $\gamma$. 
We are therefore interested in the sensitivity of
the process with respect to $c_2$ and $c_3$, and in the accuracy with which $\gamma$ can be estimated. 
We use $c_2$, $c_3$ and $\gamma_0$ to denote parameters used by the generative model to create the synthetic data; we also let $c^*_2$ and $c^*_3$ denote
 the corresponding parameters used in the spectral embedding algorithms and let 
 $\gamma^*$ represent an inferred value of $\gamma_0$.
 We choose $c_2 =1$ and $c_3 = 1/3$ so that the weight of a hyperedge is inversely proportional to the number of node pairs involved.
We let $m = 50$, $K=5$,  $a = 0.05$, and vary the decay parameter $\gamma_0$ from 0 to 10. 
Figure~\ref{fig:W2_linear} shows an example of the dyadic adjacency matrix, $W^{[2]}$, with $\gamma_0 = 4$, where dots represent non-zeros. A corresponding triadic adjacency matrix, $W^{[3]}$,
is shown in 
Figure~\ref{fig:W3_linear}. In all our tests we
discard hypergraphs that do not satisfy Assumption~\ref{assumption:connected}.

For each synthetic hypergraph, we estimate the maximum log-likelihood assuming a linear or a periodic structure using Algorithm \ref{algo:model_comparison}. For each input decay parameter $\gamma_0$, 40 hypergraphs are generated independently and the average maximum log-likelihood is plotted in Figure~\ref{fig:lnP_linear}. The shaded regions represent the estimated $80\%$ confidence interval. In this case, the linear model correctly achieves a higher average maximum log-likelihood. The tight bound of the confidence interval suggests that the result is consistent across random trials.

We then perform K-means clustering using the periodic and linear embeddings assuming 5 clusters and plot the Adjusted Rand Index (ARI) \cite{RandIndex, hubert1985comparing, chris2022ari} in  \ref{fig:rand_linear}. 
Here, a larger ARI indicates a better clustering result .
The dotted line shows the average over 40 independently trials for each $\gamma_0$ value and the shaded area is the estimated $80\%$ confidence interval. The plot suggests that the clustering from the linear embedding outperforms the clustering from
the periodic embedding.

We are interested in the effect of parameters $c_3$ and $c_3^*$ that control the weight of triadic edges in the random graph model and spectral embedding algorithm respectively. To conduct an  experiment, we fix the weight of dyadic edges $c_2 = 1$, $c_2^*=1$, and decay parameter $\gamma_0 = \gamma^*= 1$, while varying $c_3$ and $c_3^*$. The maximum log-likelihood of the linear model (Figure \ref{fig:lnP_heat_linear}) and the ARIs using the linear embedding (Figure \ref{fig:Rand_heat_linear}) are shown as heat-maps  over $c_3$ and $c_3^*$. Values are the average over 40 random trials. Overall, choosing $c_3^* = c_3$, gives the highest maximum likelihood. Therefore, when the true $c_3$ is not known, it could be estimated using a maximum likelihood method. In terms of the clustering result we note that when $c_3$ is large, for example, when $c_3 > 0.3$, using information from triadic edges by setting $c_3^*>0$ achieves a better ARI than using only diadic edges, i.e., $c_3^*=0$. This is because a large $c_3$ encourages more triadic edges to be formed within clusters, whereas a small $c_3$ leads to more triadic edges between clusters. In general the larger the $c_3$, the less sensitive the ARI is to the choice of $c_3^*$

\begin{figure}
\centering
\begin{subfigure}{.4\textwidth}
  \centering
  \includegraphics[width=\linewidth]{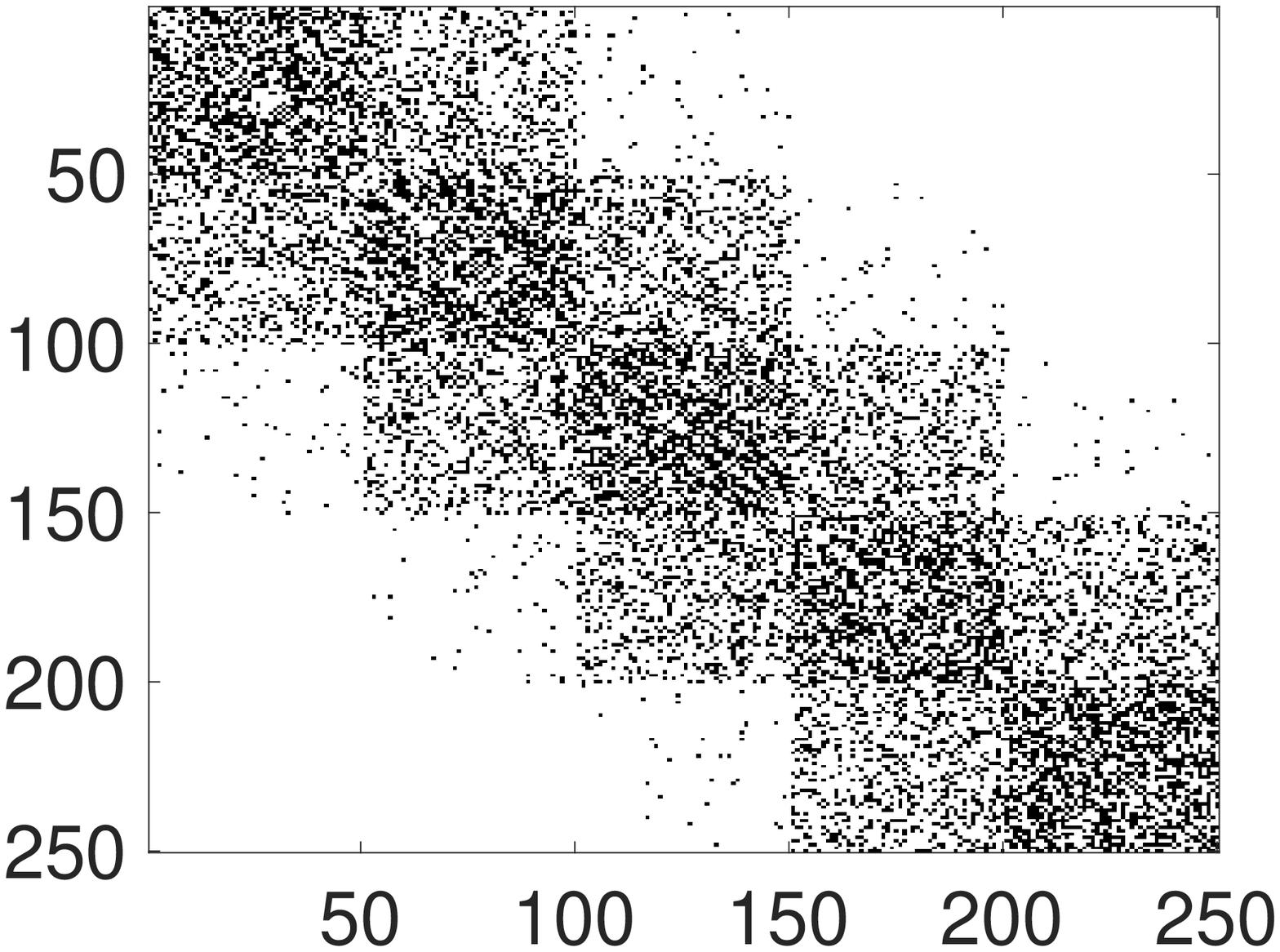}
  \caption{}
  \label{fig:W2_linear}
\end{subfigure}%
\hspace{1cm}
\begin{subfigure}{.4\textwidth}
  \centering
  \includegraphics[width=\linewidth]{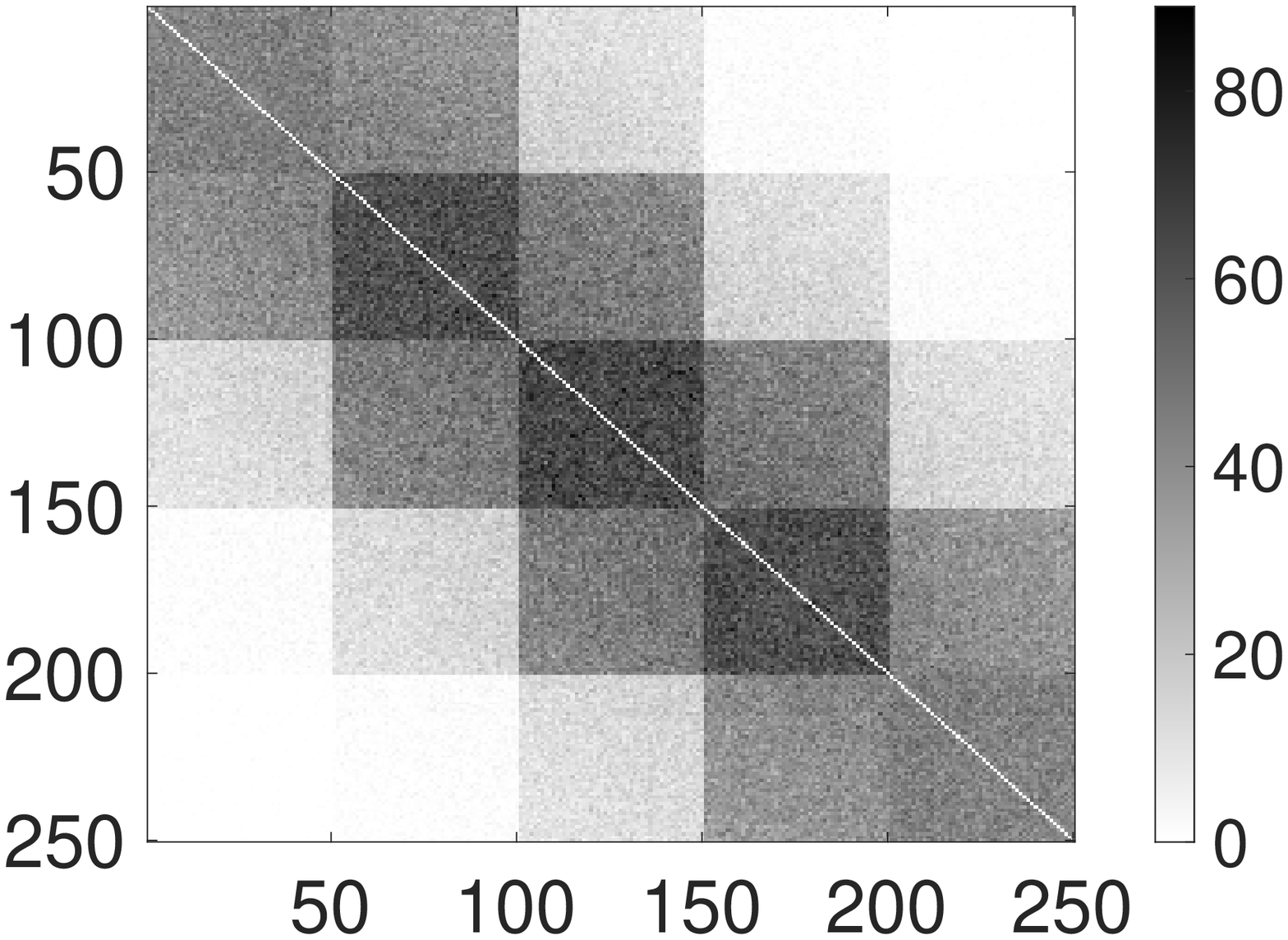}
  \caption{}
  \label{fig:W3_linear}
\end{subfigure}%
\hspace{1cm}
\begin{subfigure}{.4\textwidth}
  \centering
  \includegraphics[ width=\linewidth]{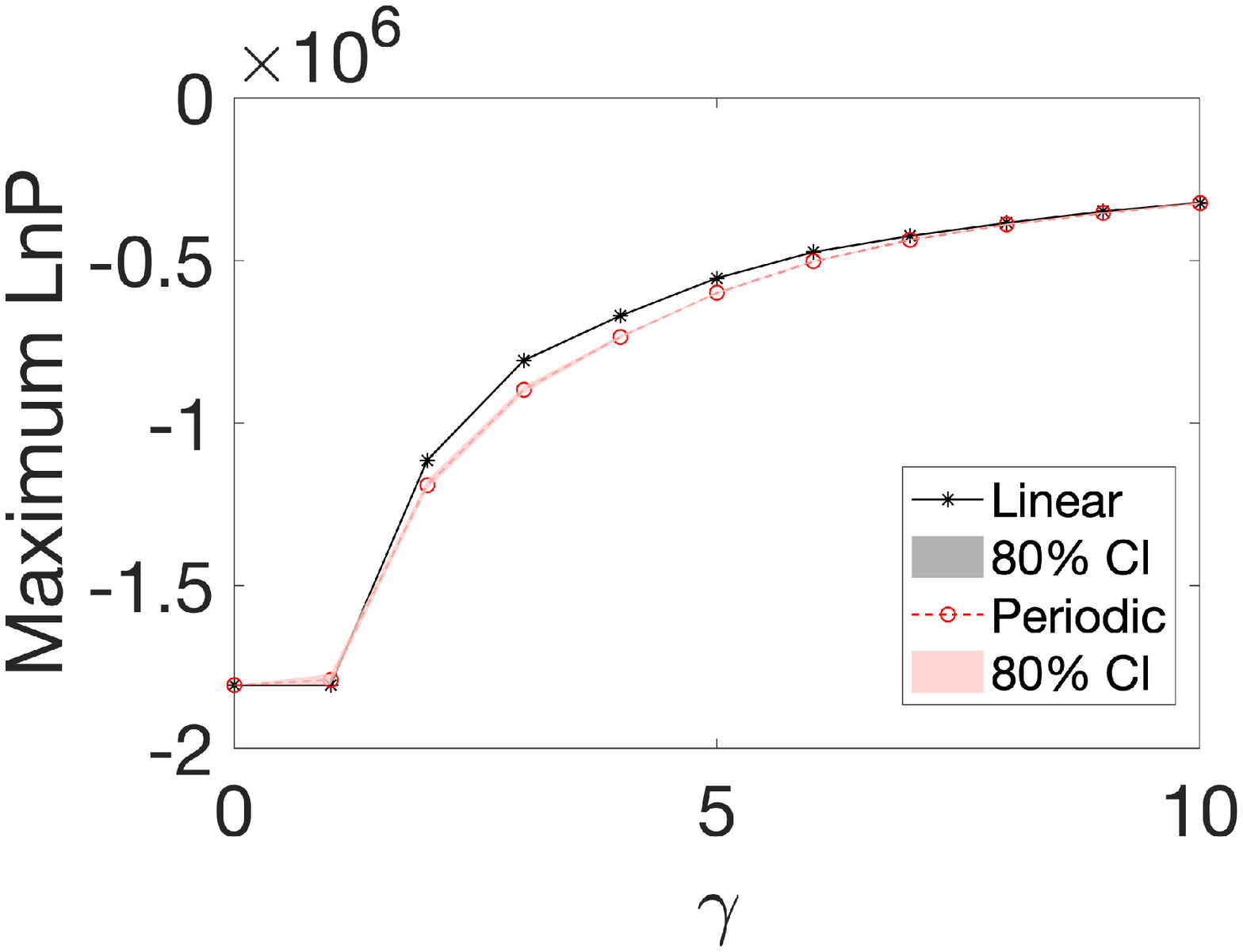}
  \caption{}
  \label{fig:lnP_linear}
\end{subfigure}
\hspace{1cm}
\begin{subfigure}{.4\textwidth}
\centering
  \includegraphics[width=\linewidth]{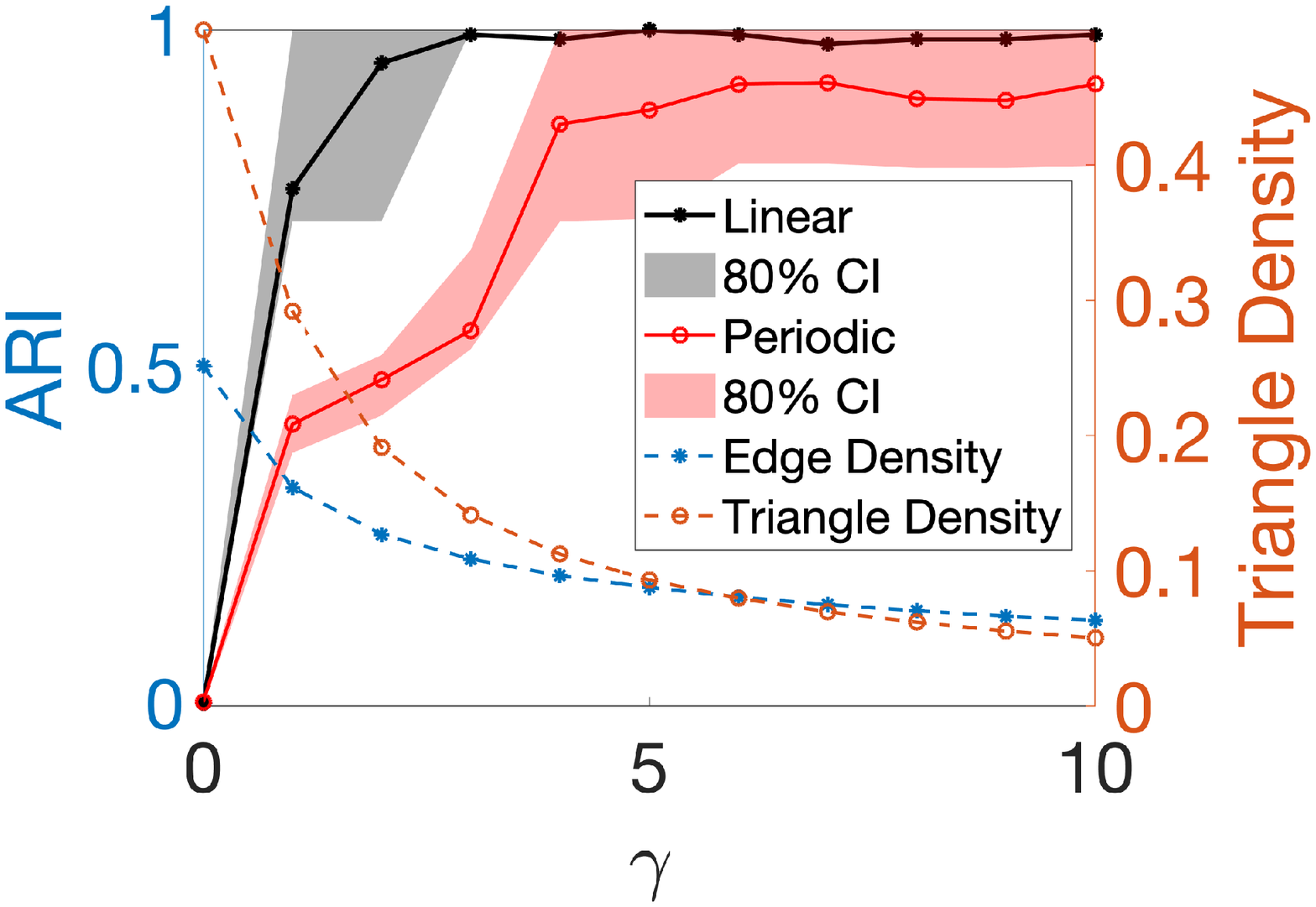}
    \caption{}
    \label{fig:rand_linear}
\end{subfigure}
\bigskip
\begin{subfigure}{.4\textwidth}
  \centering
  \includegraphics[width=\linewidth]{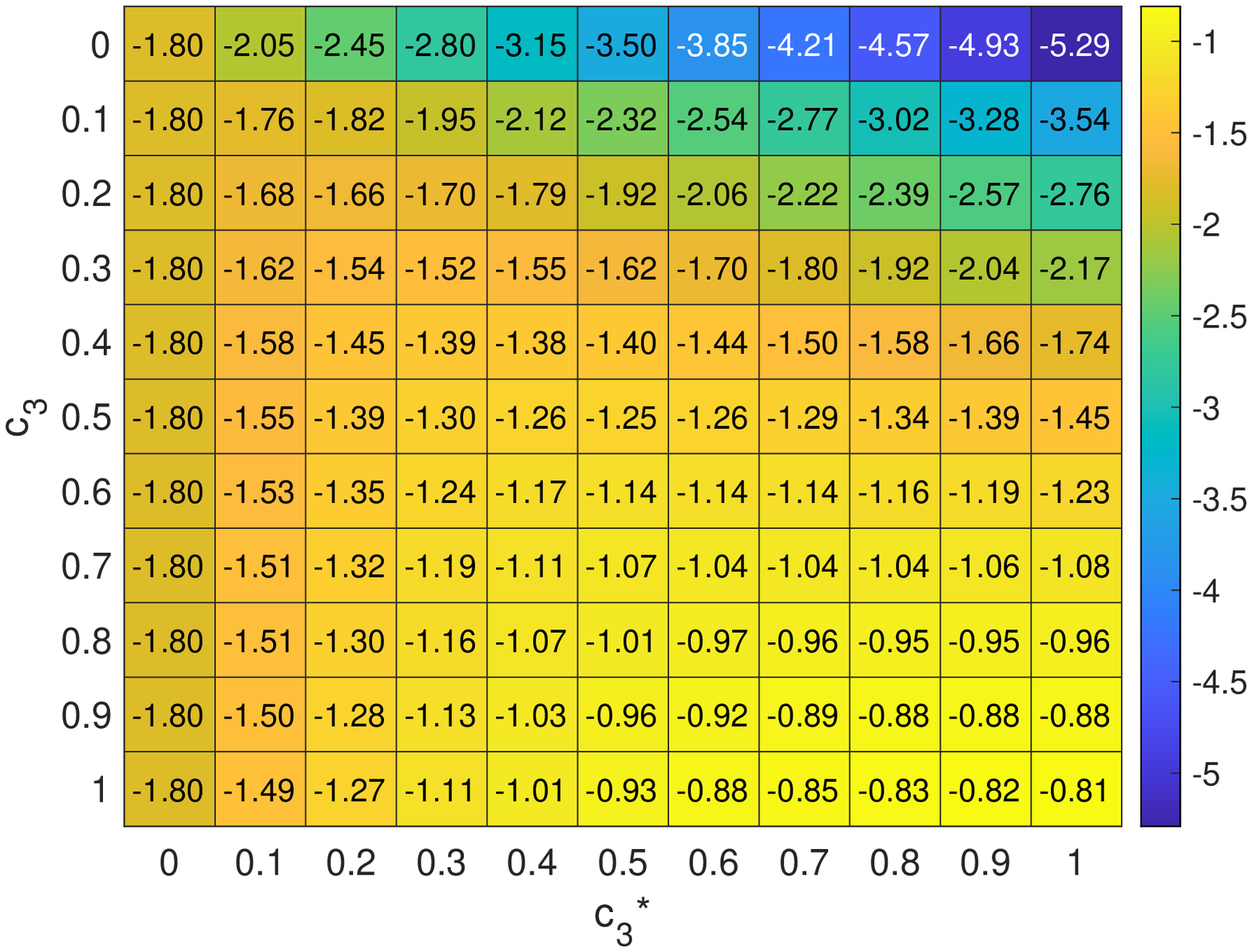}
  \caption{Linear log-likelihood of graph $\times 10^{-6}$}
  \label{fig:lnP_heat_linear}
\end{subfigure}%
\hspace{1cm}
\begin{subfigure}{.4\textwidth}
  \centering
  \includegraphics[width=\linewidth]{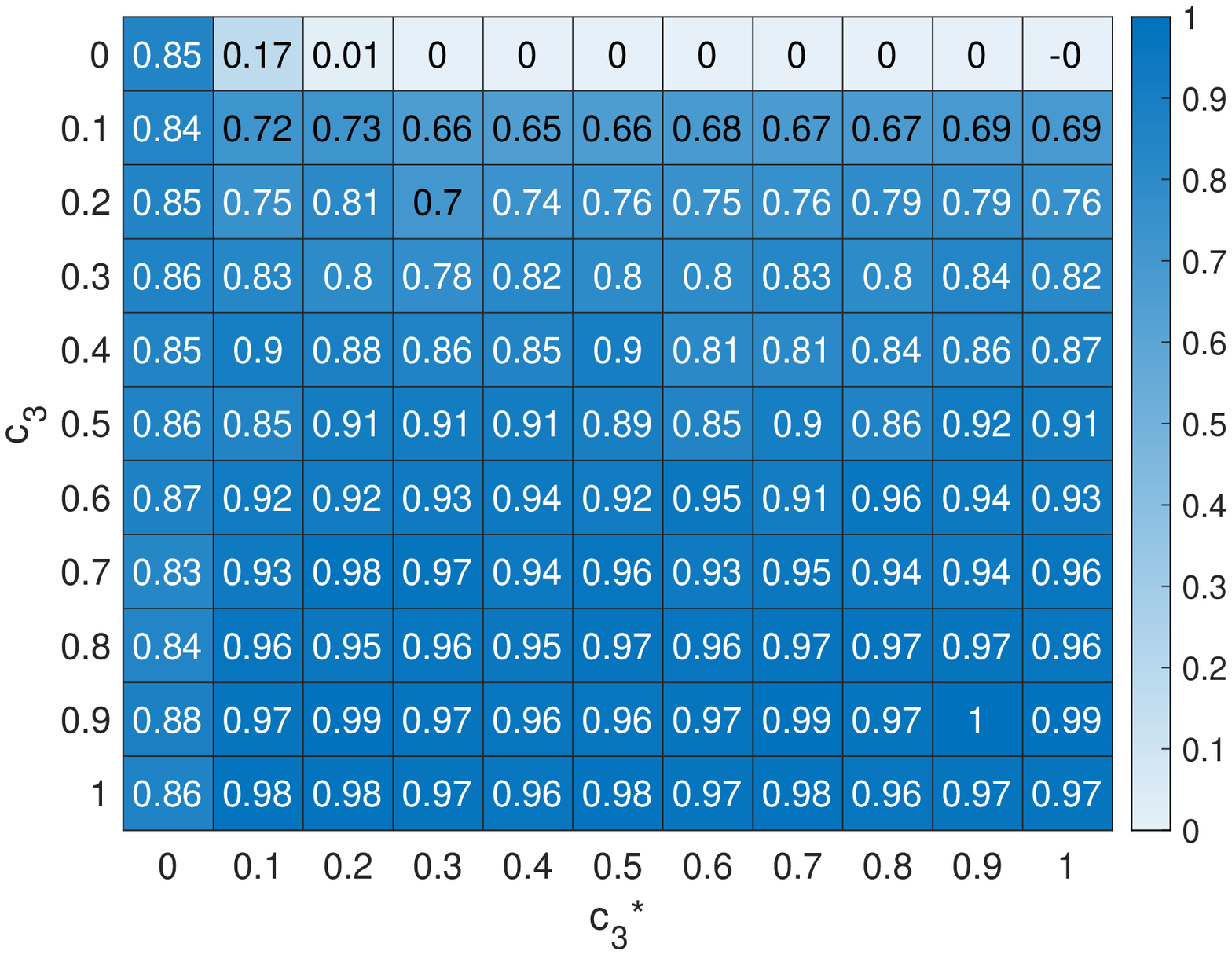}
  \caption{Linear ARI}
  \label{fig:Rand_heat_linear}
\end{subfigure}%
\caption{Model comparison experiments on synthetic linear hypergraphs. (a) Black pixels in the dyadic adjacency matrix ($\gamma_0 = 4$) represent edges, and they reveal 5 clusters in the diagonal blocks. (b) In the triadic adjacency matrix, colors reflect the number of triangles shared between nodes ($\gamma_0 = 4$). (c) The linear embedding achieves higher log-likelihoods than the periodic version for hypergraphs generated with different $\gamma_0$. (d) Adjusted Rand Indices of K-means clustering based on the linear and periodic embedding are plotted against $\gamma_0$. (e) Values in the heatmap represent the maximum likelihoods of the linear model $\times 10^{-6}$ from Algorithm \ref{algo:model_comparison}. The maxima are found along the diagonal when $c_3^* = c_3$. (f) Values represent the Adjusted Rand Indices of K-means clustering based on the linear embedding for different values of  $c_3^*$ and $c_3$. }
\label{fig:model_comparison_linear_cluster}
\end{figure}

\paragraph{Periodic hypergraph with clustered nodes}
To generate hypergraphs with periodic clusters, we use a node embedding based on a vector of angles  $\boldsymbol{\theta} = (\theta_1, \theta_2, ..., \theta_n)^T $ in $[0,2\pi)$, forming $K$ clusters $C_1, C_2, \ldots, C_K$ of size $m$. In particular, we let $\boldsymbol{\theta}_i = \frac{2 \pi (l-1)}{K}+\sigma$ if $i\in C_l$ for $1\leq l \leq K$, where $\sigma \sim \mathrm{unif}(-a,a)$ is the added noise. The hyperedges are generated using model (\ref{eq:unweighted_model}), where the incoherence function is defined in  (\ref{eq:periodic_incoherence}). We choose $a=0.05 \pi$, $c_2 =1$, $c_3 = 1/3$ and vary the decay parameter $\gamma_0$. Examples of the dyadic and triadic adjacency matrices with $\gamma_0 = 1$  are shown in Figure \ref{fig:periodic_W2} and \ref{fig:periodic_W3}.
 
  Using the same approach as in the previous section, 
 we compare the maximum log-likelihood and ARIs assuming linear and periodic structures in Figure \ref{fig:periodic_lnP} and \ref{fig:periodic_rand}. We see that the periodic model achieves a higher maximum, and on average the periodic embedding produces higher ARIs.
 
Heat-maps in Figure~\ref{fig:lnP_heatmap_periodic} and \ref{fig:rand_heatmap_periodic} show 
results for different combinations of $c_3$ and $c^*_3$ for the periodic embedding algorithm. These results were generated in the same way as 
for Figures~\ref{fig:lnP_linear} and \ref{fig:rand_linear}.
Higher maximum likelihoods are achieved near the diagonal where $c_3^* = c_3$, hence the true parameters $c_3$ for the underlying hypergraph could be estimated using the maximum likelihood method. 
As in the previous example, when $c_3 \geq 0.3$, using the triadic edges ($c_3^*$) improves the ARI. When $c_3 < 0.3$, increasing $c_3^*$ leads to an inferior clustering result. However when $c_3 \geq 0.3$, ARI becomes less sensitive to the choice of $c_3^*$ as long as it is positive. 

In summary, 
these tests indicate that the algorithms are able to 
correctly distinguish between 
linear and periodic range-dependency when 
one such structure is present in the data.
We observed that setting $c_3^* >0$ improves 
 the ARI when the triadic edges have a strong structural pattern; that is, when $c_3$ is large.  Moreover,  when the true parameter $c_3$ is unknown we recommend choosing $c_3^*$ based on a maximum likelihood estimation, that is, finding the value $c_3^*$ that returns the largest maxima in Algorithm \ref{algo:model_comparison}. Such a choice also achieve reasonable ARIs in our synthetic examples as shown in the diagonal entries in Figure \ref{fig:Rand_heat_linear} and \ref{fig:rand_heatmap_periodic}. 

\begin{figure}
\centering
\begin{subfigure}{.4\textwidth}
  \centering
  \includegraphics[width=\linewidth]{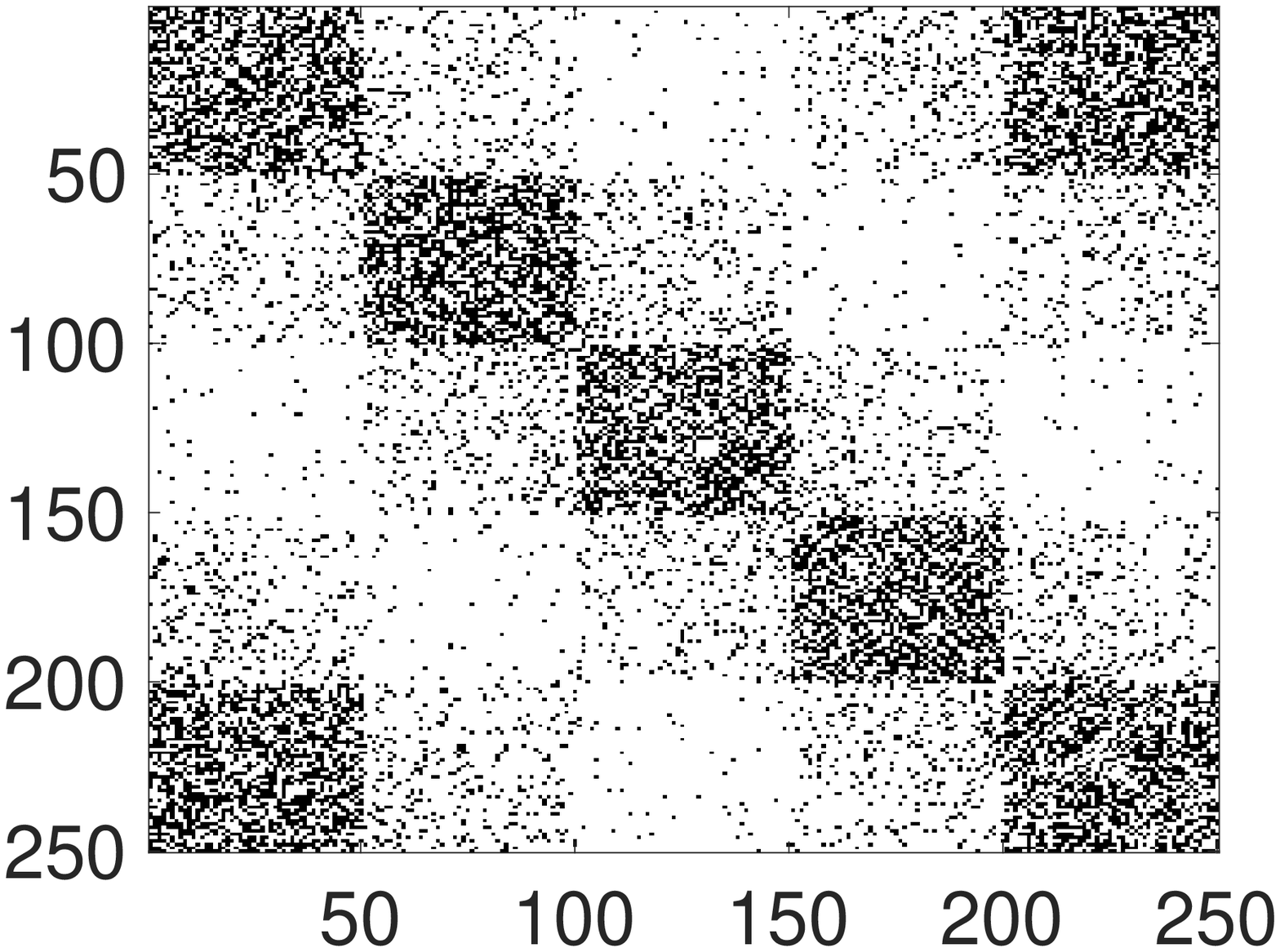}
  \caption{}
  \label{fig:periodic_W2}
\end{subfigure}%
\hspace{1cm}
\begin{subfigure}{.4\textwidth}
  \centering
  \includegraphics[width=\linewidth]{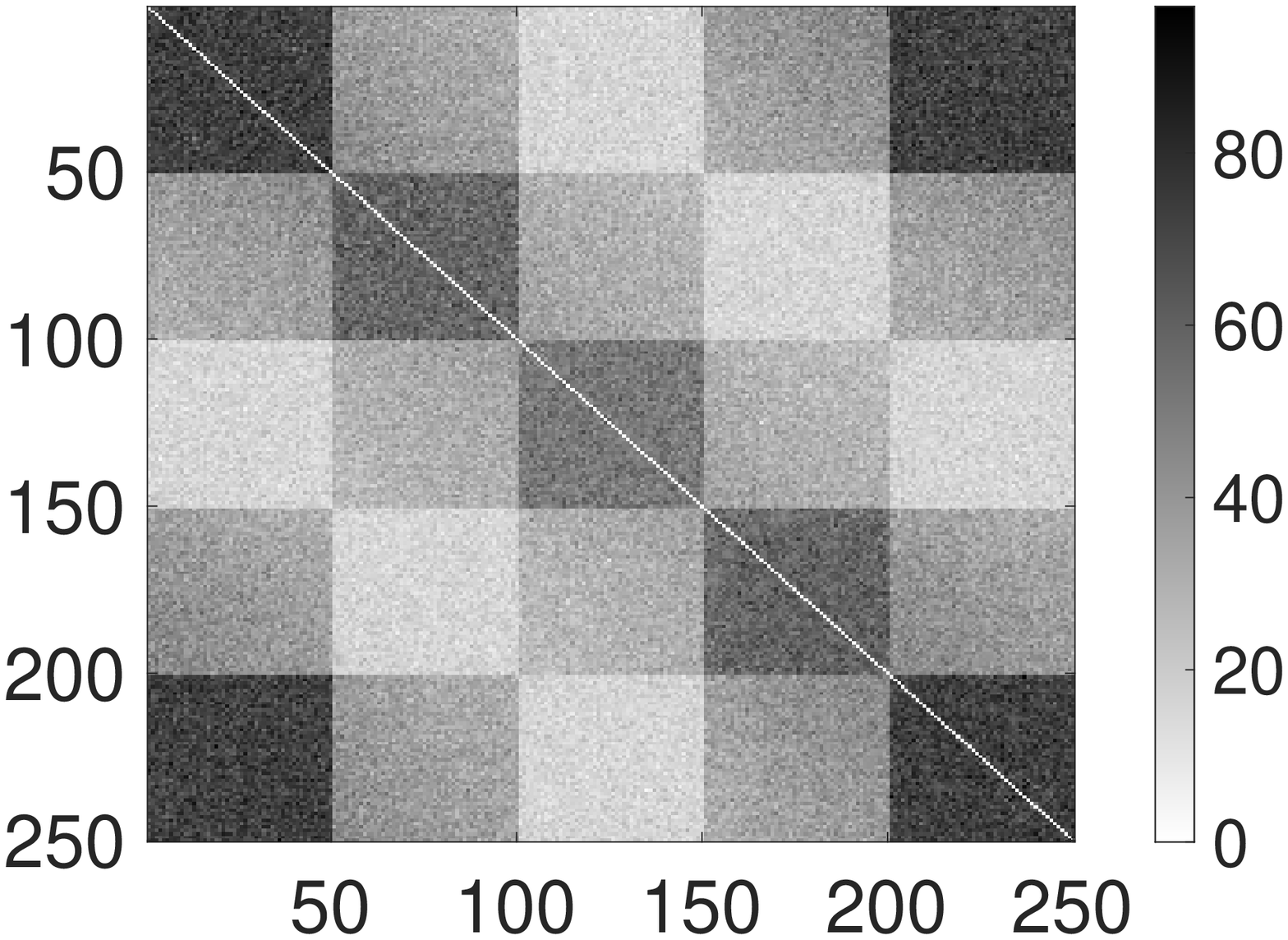}
  \caption{}
  \label{fig:periodic_W3}
\end{subfigure}%
\hspace{1cm}
\begin{subfigure}{.4\textwidth}
  \centering
  \includegraphics[width=\linewidth]{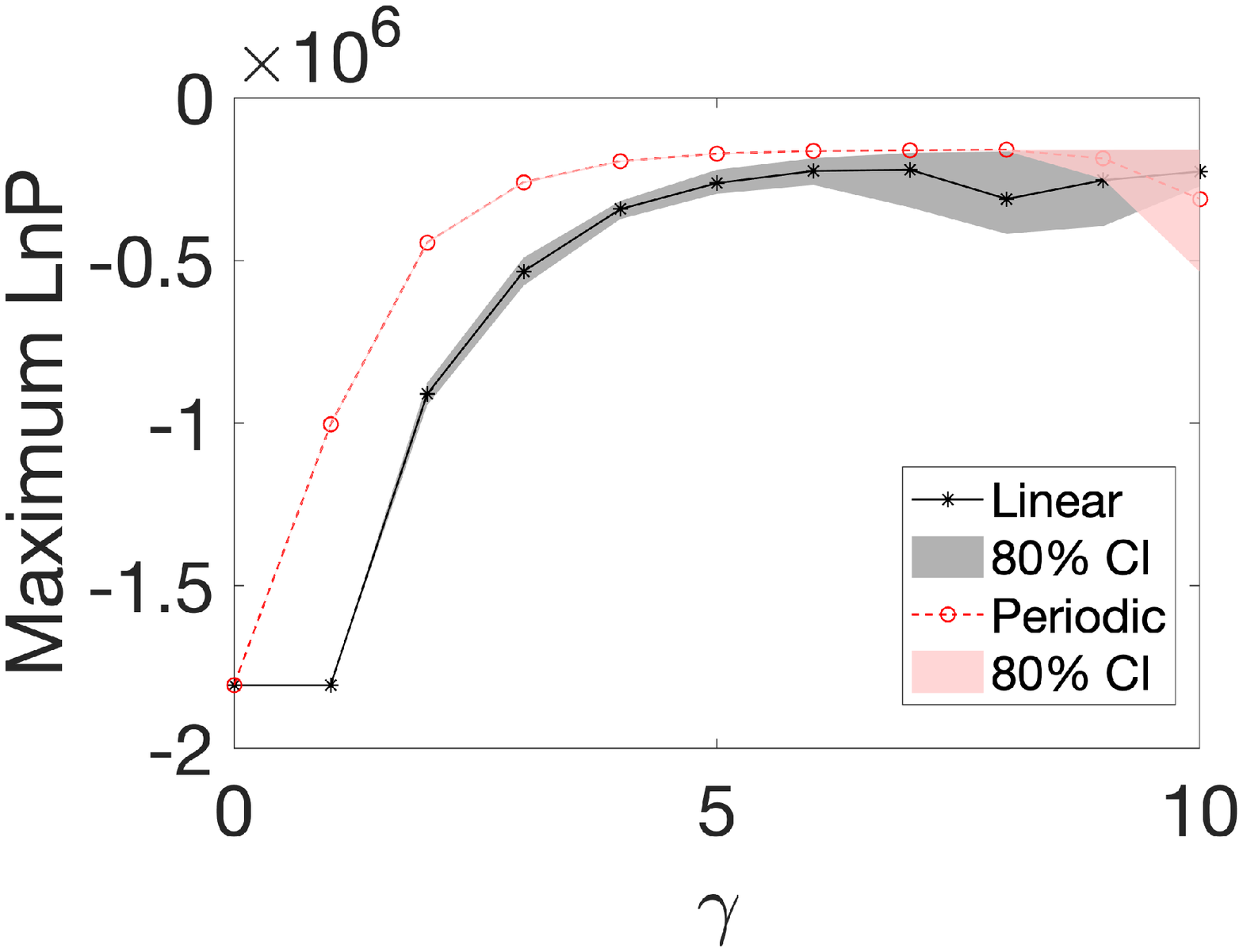}
  \caption{}
  \label{fig:periodic_lnP}
\end{subfigure}
\begin{subfigure}{.4\textwidth}
\centering
  \includegraphics[width=\linewidth]{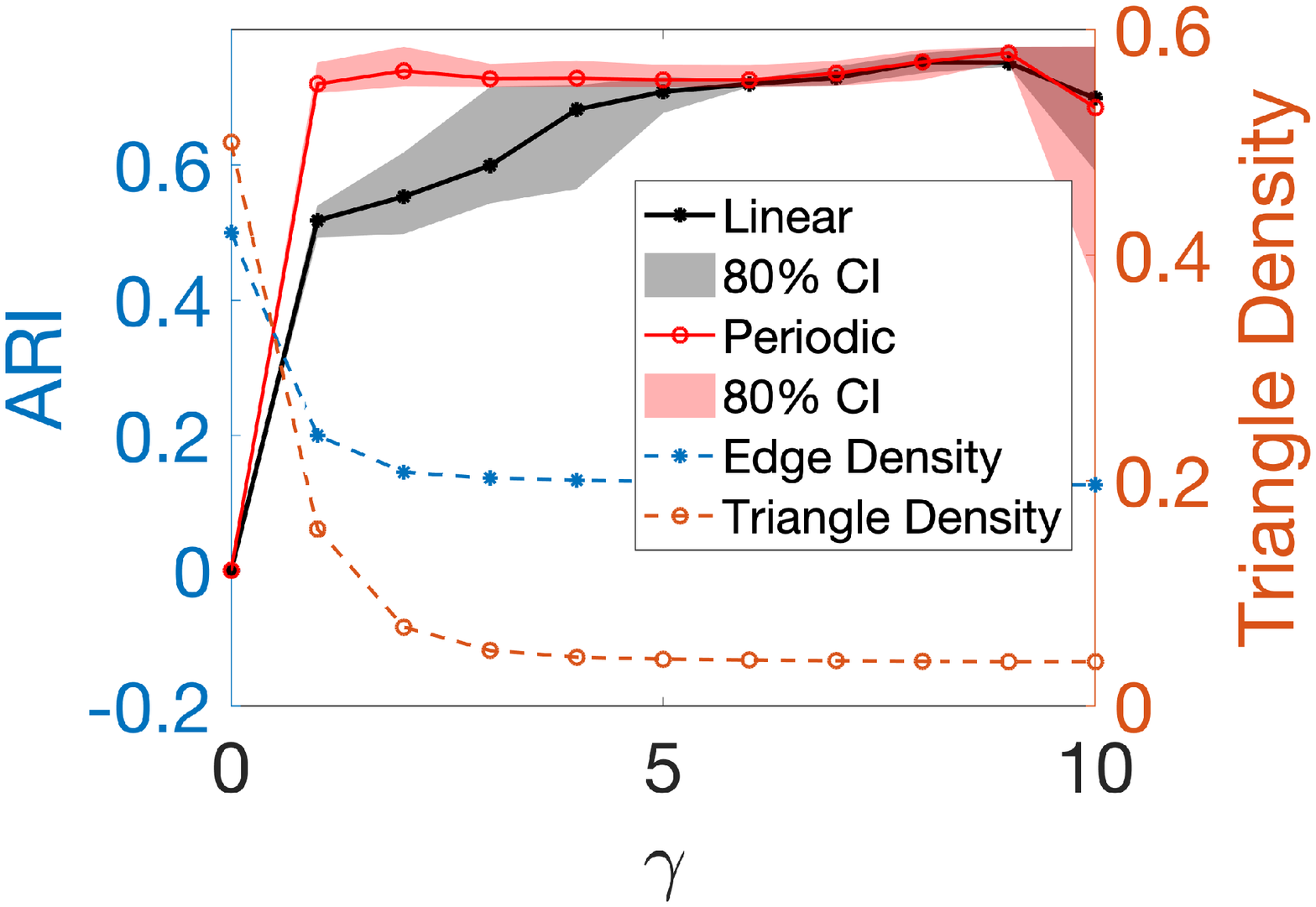}
    \caption{}
    \label{fig:periodic_rand}
\end{subfigure}
\bigskip
\begin{subfigure}{.4\textwidth}
  \centering
  \includegraphics[width=\linewidth]{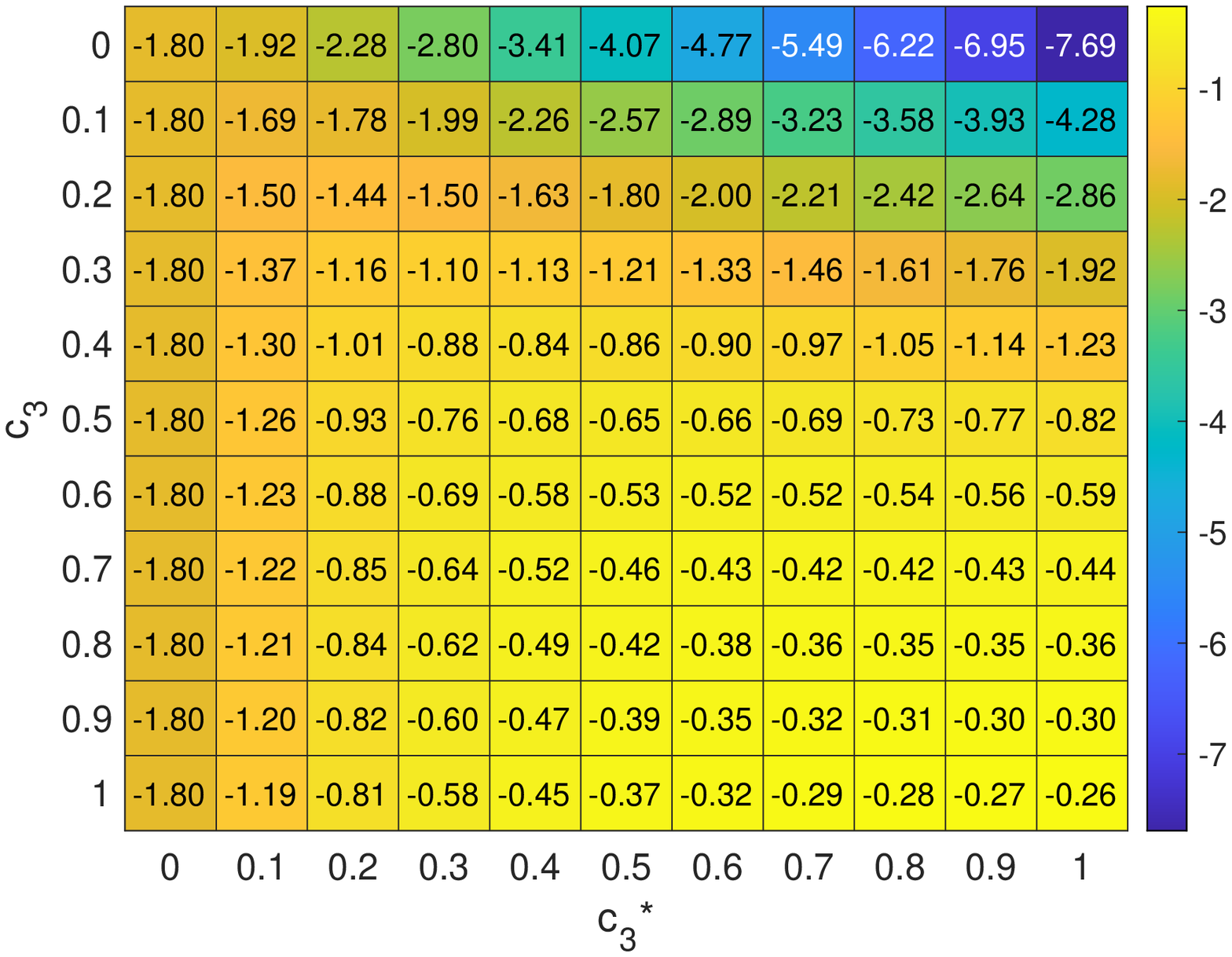}
  \caption{Log-likelihood of periodic model $\times 10^{-6}$}
  \label{fig:lnP_heatmap_periodic}
\end{subfigure}%
\hspace{1cm}
\begin{subfigure}{.4\textwidth}
  \centering
  \includegraphics[width=\linewidth]{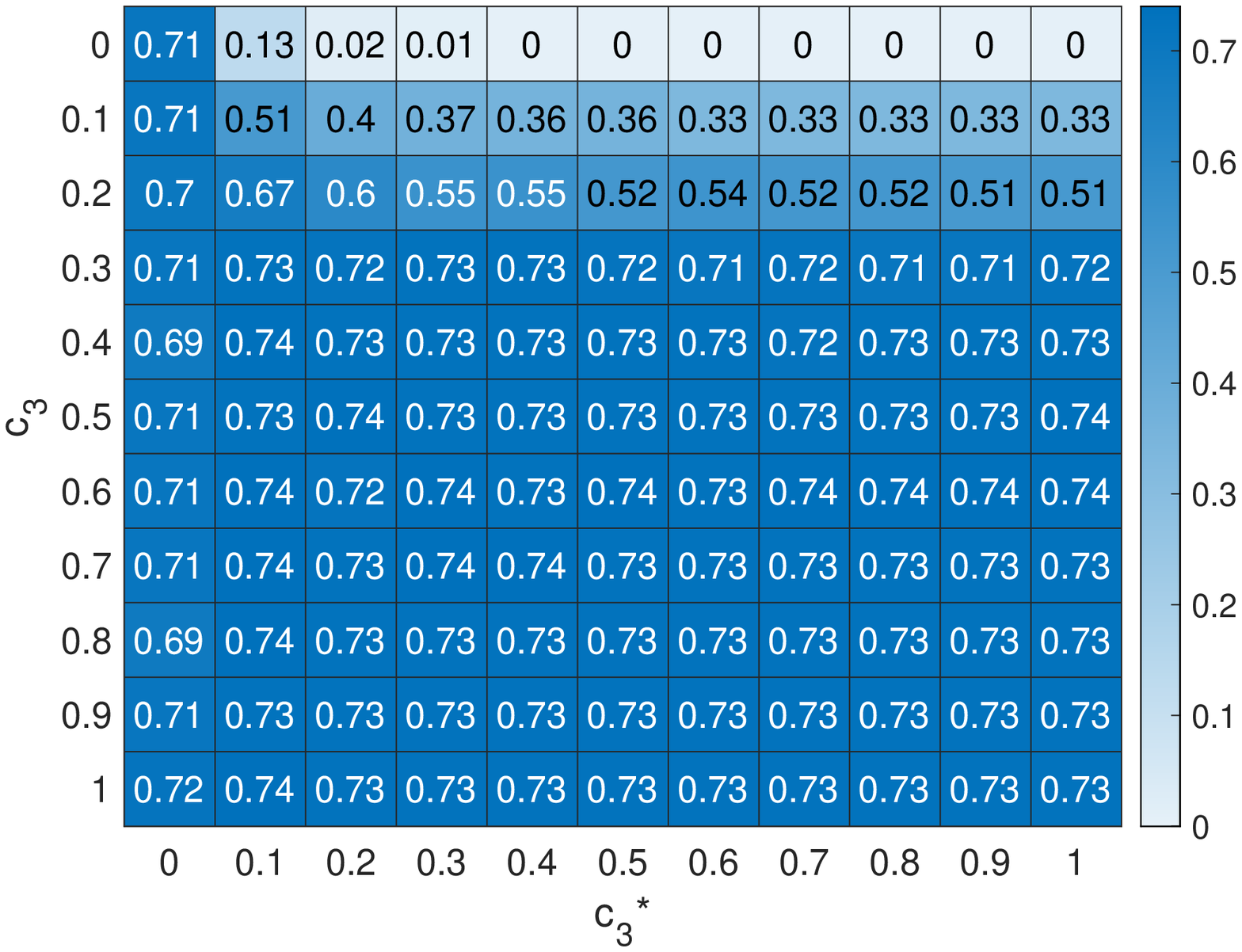}
  \caption{Periodic ARI}
  \label{fig:rand_heatmap_periodic}
\end{subfigure}%
\caption{Model comparison experiments on synthetic periodic hypergraphs. (a) The bottom-left and top-right blocks in the dyadic adjacency matrix ($\gamma_0 = 1$) reveal a periodic clustered structure. (b) Similar periodic clustered pattern is shown in the triadic adjacency matrix  ($\gamma_0 = 1$). (c) The periodic embedding achieves higher log-likelihoods than the linear embedding for different values of $\gamma_0$. (d) Adjusted Rand Indices of K-means clustering based on the linear and periodic embedding are plotted against $\gamma_0$. (e) Values in the heatmap indicate the maximum likelihoods of the periodic model $\times 10^{-6}$ from Algorithm \ref{algo:model_comparison} and the maxima lie on the diagonal where $c_3^* = c_3$. (f) Values represent the Adjusted Rand Indices of K-means clustering based on the periodic embedding for different values of  $c_3^*$ and $c_3$.}
\label{fig:model_comparison_periodic_cluster}
\end{figure}

\subsubsection*{Real Hypergraphs}
\label{subsec:real_hypergraphs}

\paragraph{High School Contact Data}
The high school contact data from  \cite{mastrandrea2015contact} records the frequency of student interaction. Students are represented as nodes, and contacts between two or three students are registered as dyadic or triadic edges. We retrieved the hypergraph from \cite{Chodrow21} containing 327 nodes, and only studied its  dyadic and triadic edges considering the computational complexity. We construct the hypergraph Laplacian $L = c_2^* L^{[2]} + c_3^* L^{[3]}$ and perform linear and periodic spectral embedding. For the linear embedding we map nodes into 3-dimensional Euclidean space using the eigenvectors corresponding to the three smallest eigenvalues that are larger than $0.01$. We make this choice because the eigenvector associated with the smallest non-zero eigenvalue has only a few non-zero entries and leads to trivial clusters. We fix $c_2^* = 1$ and vary $c_3^*$ since only the relative weight $c_3^*/c_2^*$ matters in node embedding. 

The  the maximum likelihoods and ARIs evaluated using various $c_3^*$ are shown in the left column in Figure~\ref{fig:model_comparison_real}. The true clusters are defined by the classes the students came from. Overall the periodic embedding achieves higher likelihoods and ARIs despite the linear embedding involving more parameters. Since linear clusters tend to have more marginalized groups that are far from other clusters, our results may suggest a lack of marginalisation driven by class membership.

We note that setting $c^*_3=0$ causes the algorithm to ignore triangles, and hence to reduce to classical spectral clustering. For the linear algorithm, we see that incorporating triadic edges by using a positive $c^*_3$ can improve the ARI by up to around $0.09$. We note that in \cite{Chodrow21}, modularity maximization-based clustering achieved ARI=1 on the same data. However, those methods have more parameters, which makes the ARI not directly comparable.



\paragraph{Primary School Contact Data}
The primary school contact hypergraph \cite{Chodrow21} is constructed from the contact pattern between primary students from 10 classes \cite{Stehl-2011-contact}. Nodes represent students or teachers, and hyperedges represent their physical contact. Each node is labelled by the class of the students or as a teacher. The hypergraph contains 242 nodes and 11 classes of labels. We extracted the dyadic and triadic edges from the hypergraph and performed likelihood comparison and clustering with four eigenvectors associated with the four smallest eigenvalues that are greater than 0.01. The middle column in Figure~\ref{fig:model_comparison_real} suggests the periodic embedding achieves the maximum likelihood at $c_3^* = 0.7$ and overall performs better than the linear embedding in the clustering task.  These results may be related to the existence of a teacher group that connects with all student groups. When we arrange dyadic and triadic adjacency matrices by node classes, these connections will appear as off-diagonal entries. As we have shown in Figure \ref{fig:periodic_W2} and \ref{fig:periodic_W3}, the periodic model tends to produce more off-diagonal connections than the linear model.



\paragraph{Senate Bills Data}
In the senate bills hypergraph\cite{Fowler-2006-connecting, Fowler-2006-cosponsorship, Chodrow21}, nodes are US Congresspersons and hyperedges are the sponsor and co-sponsors of Senate bills. There are in total 294 nodes, and each node is labelled as either Democrat or Republican. We performed likelihood comparison and clustering with only the dyadic and triadic edges. Since the node degree distribution is highly inhomogeneous, we observe many trivial eigenvectors that are close to indicator functions. To address this issue we trimmed off the top and bottom 2\% nodes by node degree, and use the eigenvector associated with the smallest eigenvalue that is greater than 0.01. The linear and periodic models have similar maximum likelihoods and clustering ARIs, as shown in the right column in Figure~\ref{fig:model_comparison_real}. In contrast with previous examples, there are only two clusters present in this data set. Hence the difference between the periodic and linear models, which could be reflected in the connection (or disconnection) pattern between the first and the last group, is less prominent.

\begin{figure}
\centering
\begin{subfigure}{.3\textwidth}
  \centering
  \includegraphics[width=\linewidth]{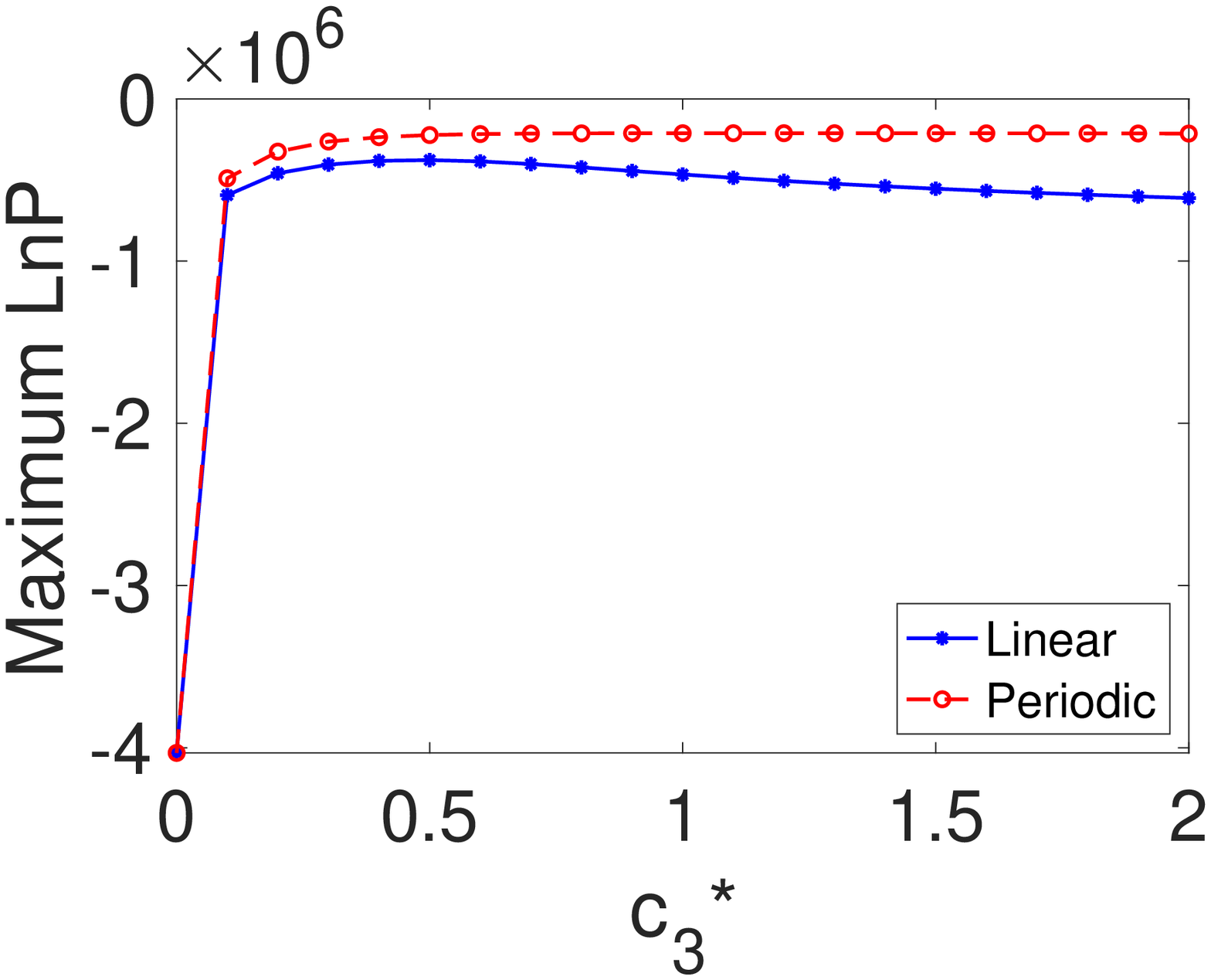}
  \caption*{}
\end{subfigure}
\hspace{0.3cm}
\begin{subfigure}{.3\textwidth}
  \centering
  \includegraphics[width=\linewidth]{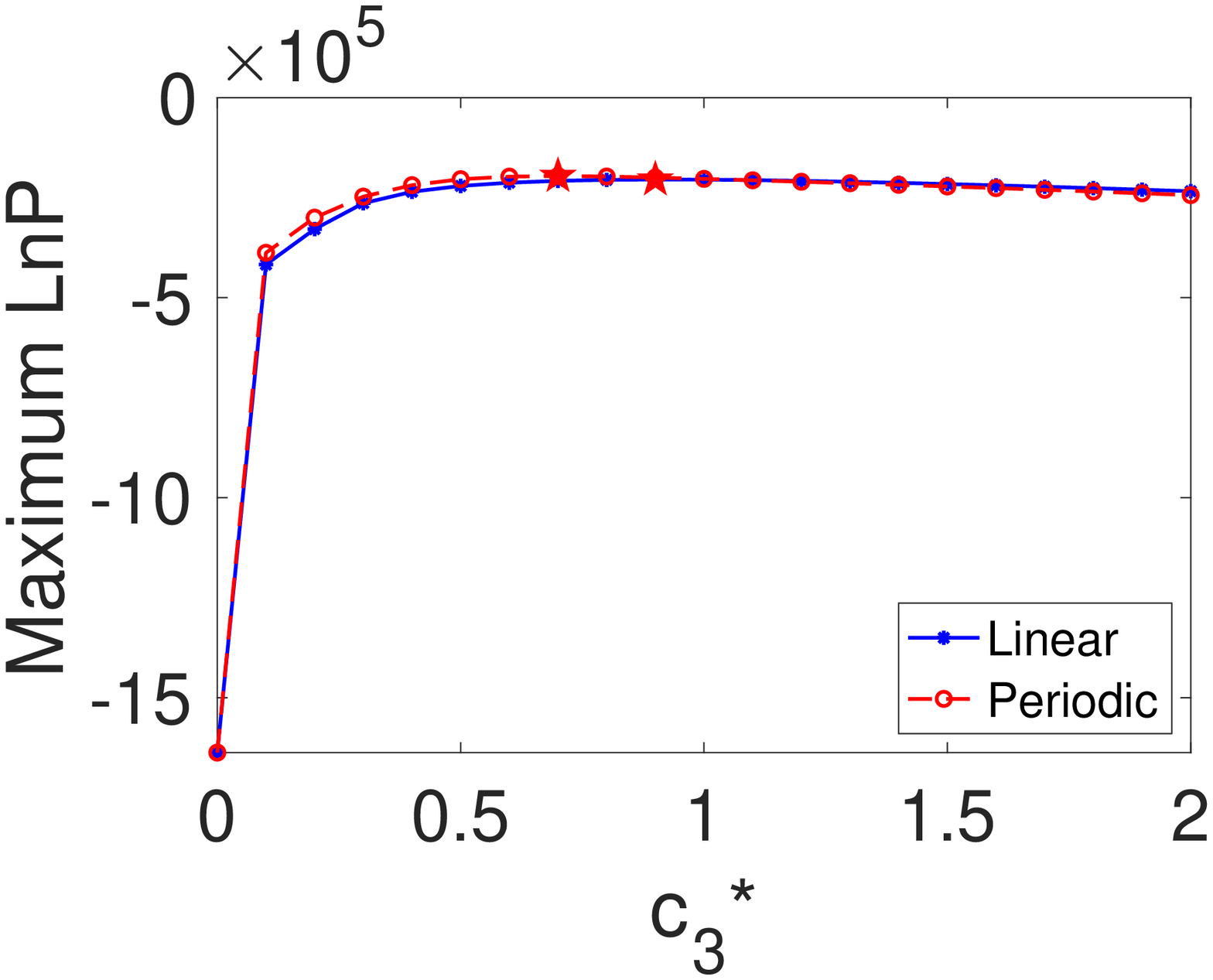}
  \caption*{}
\end{subfigure}
\hspace{0.3cm}
\begin{subfigure}{.3\textwidth}
  \centering
  \includegraphics[width=\linewidth]{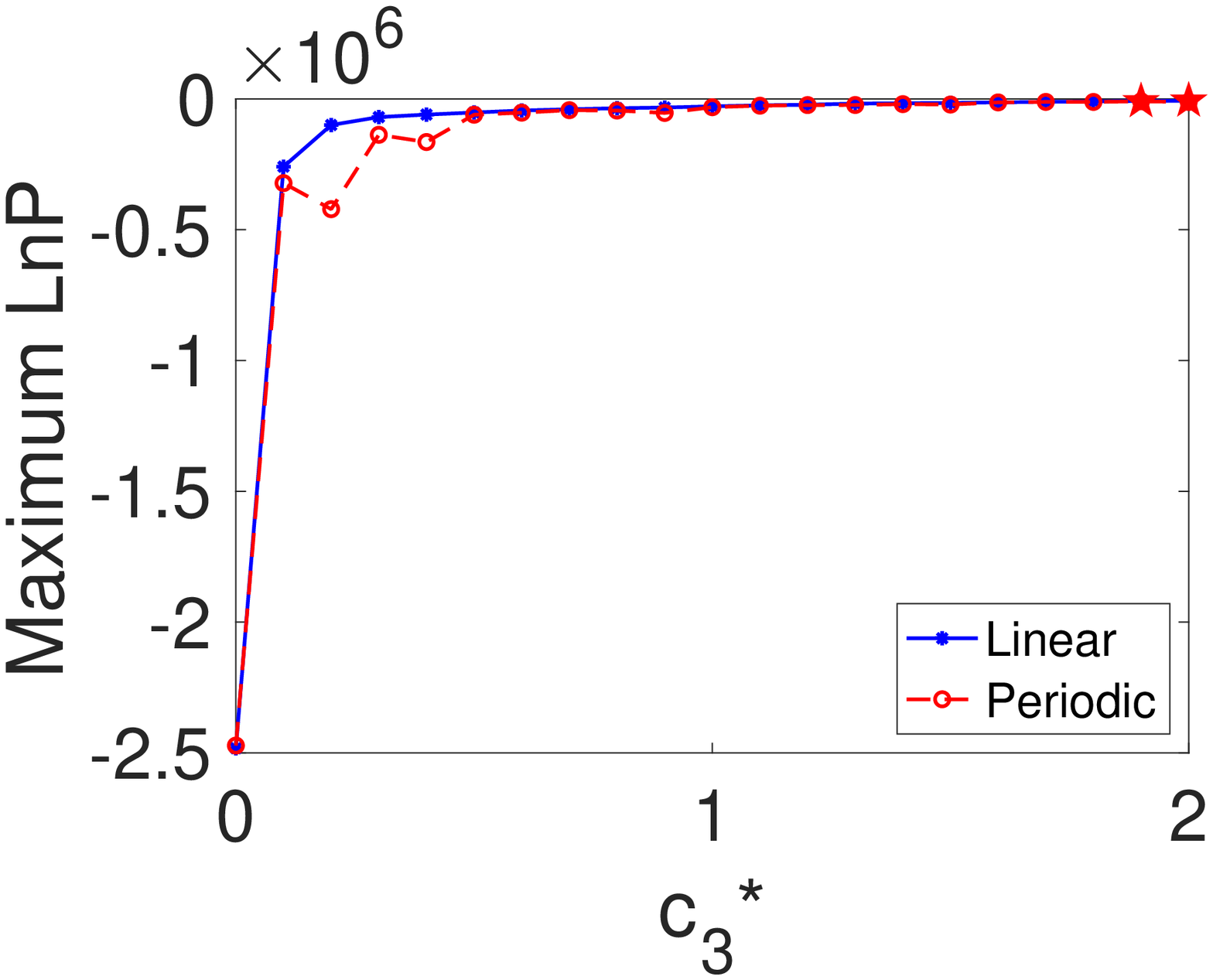}
  \caption*{}
\end{subfigure}
\begin{subfigure}{.3\textwidth}
  \centering
  \includegraphics[width=\linewidth]{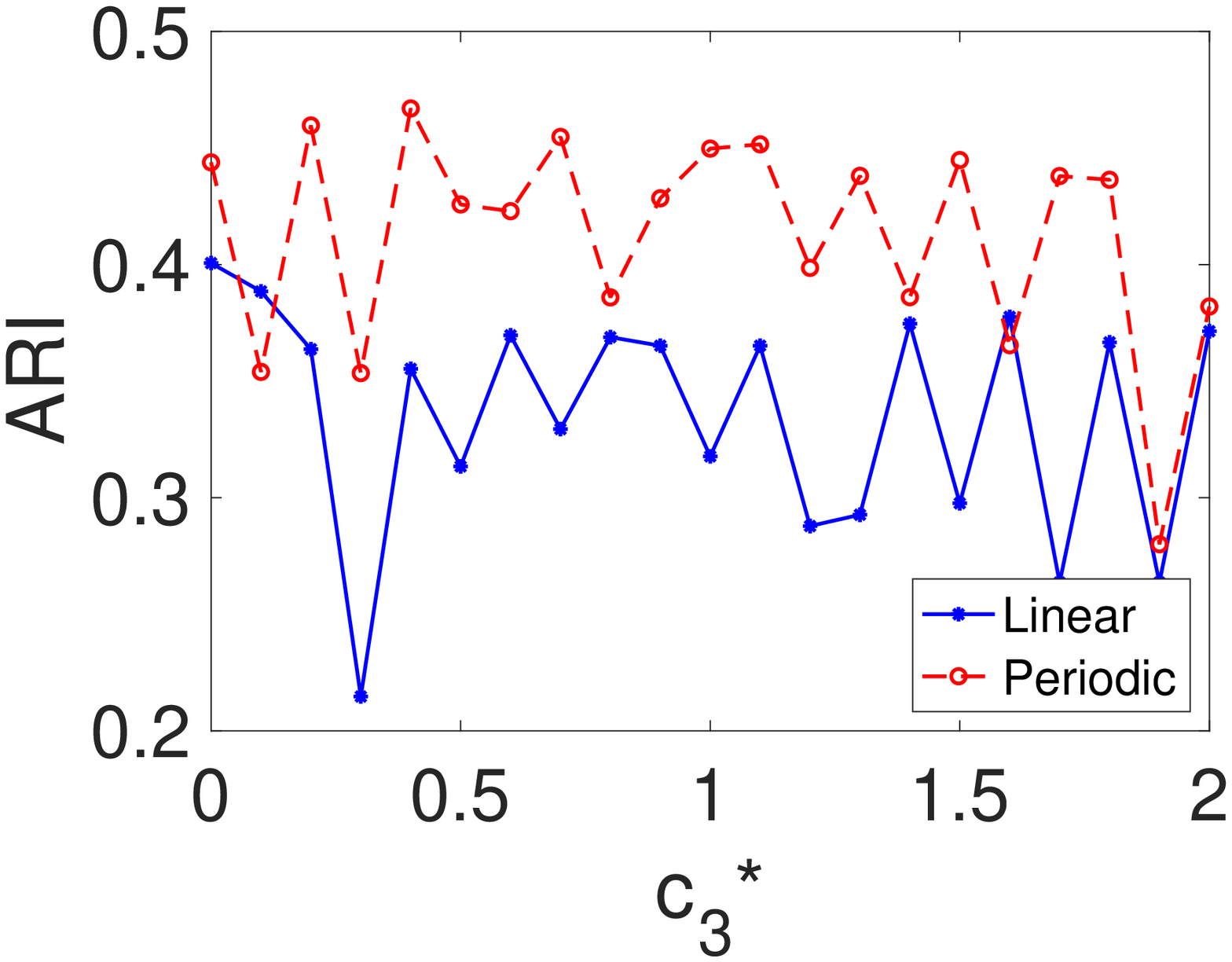}
  \caption*{High School Contact}
\end{subfigure}
\hspace{0.3cm}
\begin{subfigure}{.3\textwidth}
  \centering
  \includegraphics[width=\linewidth]{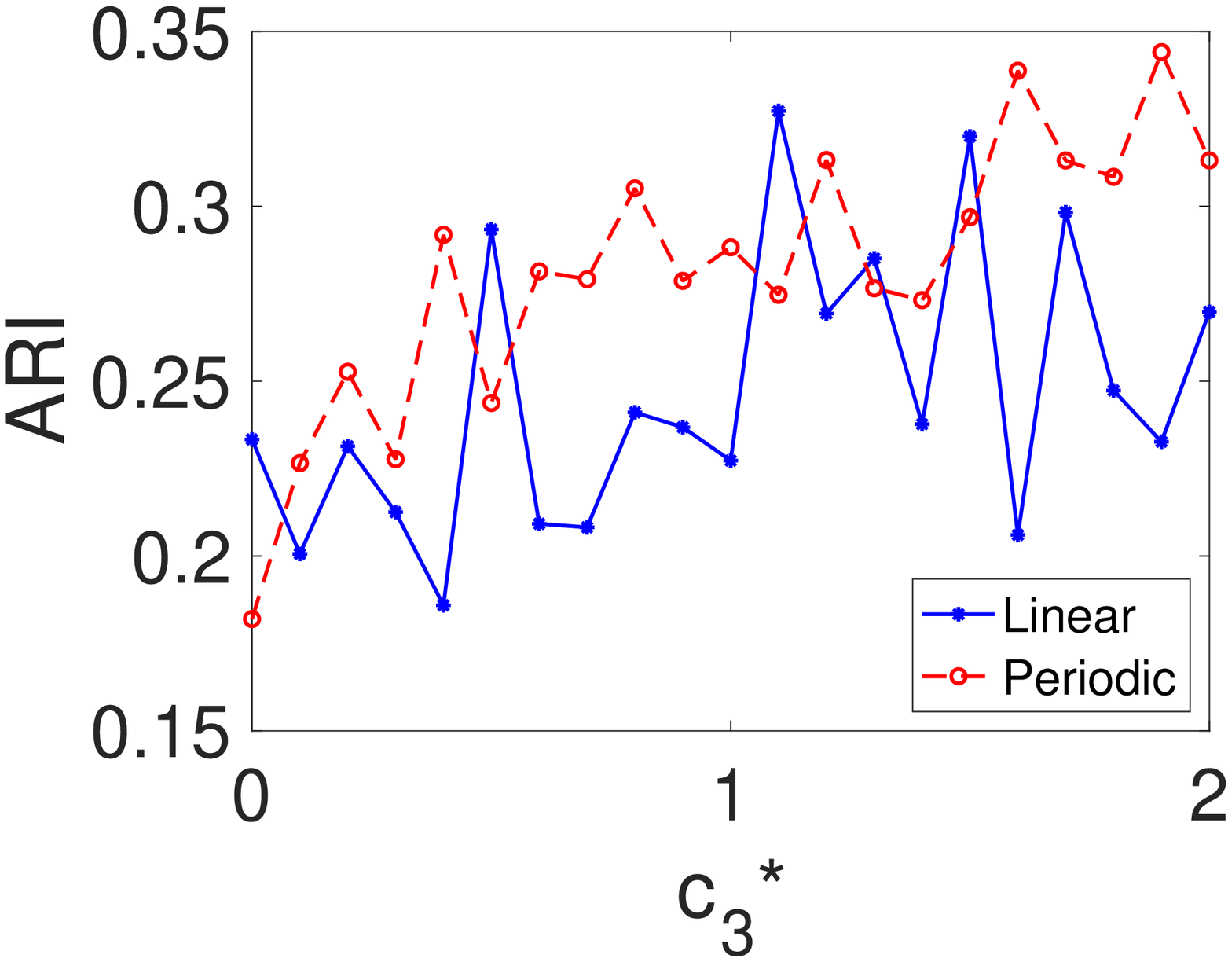}
  \caption*{Primary School Contact}
\end{subfigure}%
\hspace{0.3cm}
\begin{subfigure}{.3\textwidth}
  \centering
  \includegraphics[width=\linewidth]{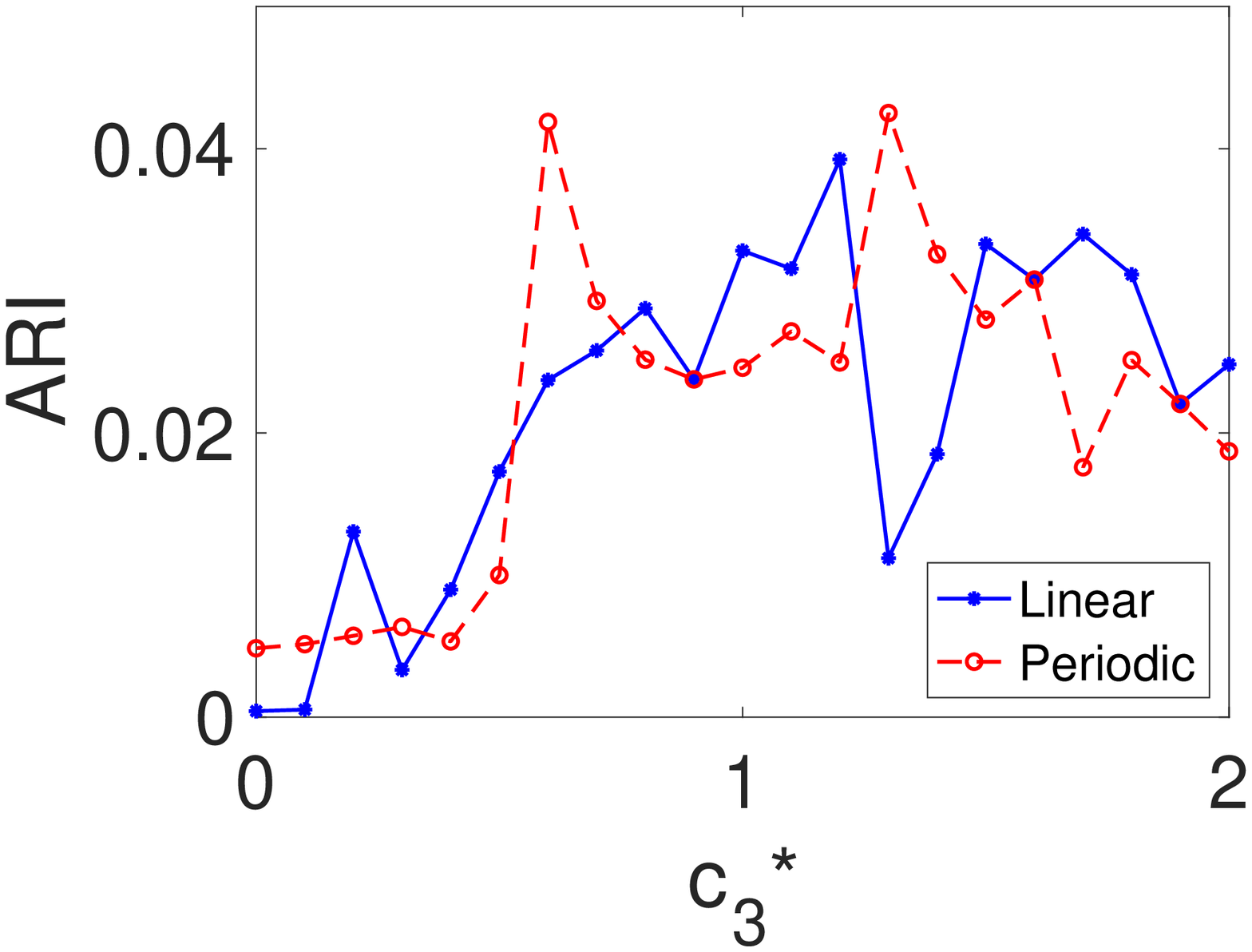}
  \caption*{Senate Bills}
\end{subfigure}%

\caption{Model comparison experiments on the high school contact data (left), the primary school contact data (middle), and the senate bill data (right). Top row shows the maximum log-likelihoods of the linear and periodic embedding for different values of $c_3^*$ where pentagrams indicate maxima. Bottom row plots the Adjusted Rand Indices of K-means clustering based on the linear and periodic embedding for various $c_3^*$.}
\label{fig:model_comparison_real}
\end{figure}

\subsection*{Hyperedge Prediction}

Once the node embeddings are estimated from the spectral algorithms, the probability of hyperedges may be computed from the proposed models. The hyperedge probability can naturally serve as a score for hyperedge prediction.  We implement and test such triadic edge prediction on timestamped high school contact data \cite{benson2018simplicial, mastrandrea2015contact} , primary school contact data \cite{Stehl-2011-contact, benson2018simplicial}, and synthetic linear hypergraphs. The results will be compared against approaches based on average-scores proposed in \cite{benson2018simplicial}. Other hyperedge prediction methods include feature-based prediction \cite{yoon2020much}, model-based prediction\cite{scholtes2017network}, and machine learning-based prediction \cite{yadati2020nhp}. 

For high school and primary school contact data, we used three and four eigenvectors respectively corresponding to the smallest eigenvalues that are greater than $0.01$ to be consistent with the previous section, and only consider dyadic and triadic edges.  The hyperedges are sorted by time stamps and split into training and testing data. For example, an  $80:20$ training/testing splitting ratio means we use the first $80\%$ of the hyperedges to train the model and the last $20\%$ to test the predictions. When the training ratio is low, the subgraph for training may be disconnected thus violating Assumption \ref{assumption:connected}. Therefore, we only consider nodes in the largest connected component of the graph associated with the binarized version of $L$ of the training subgraph, and test the prediction on the same set of nodes. Note that in real data, the parameters $\gamma_0$, $c_2$ and $c_3$ for the hypergraph model are unknown.  We fix $c_2^*=1$ and choose $c_3^*$ and $\gamma^*$ using a maximum likelihood estimate through a grid search on the training data. 

On the training set, we assign scores to each triplet using five methods: a random score as baseline, hyperedge probability from the linear model, arithmetic mean, harmonic mean, and geometric mean from \cite{benson2018simplicial}. On the test set, we measure the prediction performance with the area under precision-recall curve (AUC-PR) \cite{davis2006relationship}. A Precision-Recall (PR) curve traces the Precision = True Positive/ (True Positive + False Positive) and Recall = True Positive / (True Positive + False Negative) for different thresholds. The AUC-PR is a measure that balances both Precision and Recall where 1 means perfect prediction at any threshold. Setting  $c_3^* = 0$ will assign probability of $0.5$ to all triplets, which is equivalent to the random score approach if we break ties randomly.


On the high-school contact data shown in 
Table~\ref{tbl:triangle_prediction}, the harmonic and geometric mean attain the highest AUC-PR for large amounts of training data, see, for example the $80:20$ data split; while the linear model predictions achieve the best results for small amounts of training data, as seen in the $20:80$ data split. This could be because when training data is insufficient, there are more unobserved "missing" dyadic and triadic edges. In this case, the node embedding algorithm can infer node proximity based on common neighbours. In other words, it can place nodes with common neighbours nearby even if they haven't been directly linked before.  On the other hand, the geometric and harmonic mean will assign a score of zero to a triplet if none of the nodes has been connected previously, and therefore will predict no triadic edges.

We also test the triadic edge prediction on the synthetic linear hypergraphs, generated in the manner described in subsection~\ref{subsec:synthetic_linear}, with $K = 4$, $m=60$, $\gamma = 10$, $c_2 = 1$, and $c_3 = 0.3$, such that the clustered pattern resembles the high-school contact data. We consider three eigenvectors associated with the smallest eigenvalues that are greater than 0.01 for the synthetic linear model. Since defining a periodic model with more than one eigenvector is beyond the scope of this work, we only test the linear hypergraphs. We randomly select a portion of the hyperedges as the training set, while ensuring the sampled hypergraph is connected, and test the performance on the rest of the hyperedges. The AUC-PR averaged over 20 random hypergraphs is shown in Table~\ref{tbl:triangle_prediction}. We observe that the linear model outperforms random score and average scores for various training data sizes.

\begin{table}
\begin{center}
\begin{tabular}{||c c c c c c||} 
 \hline
 Data (Train:Test) & Random & Linear & Arith Mean & Geo Mean & Harm Mean \\ 
 \hline
 Highschool (80:20) & 5.3e-5 &  9.3e-4   &  9.4e-4 & \textbf{6.2e-3} & 6.0e-3  \\ 
 \hline
  Highschool (60:40) & 1.2e-4 &  1.2e-3   & 2.1e-3 & \textbf{1.1e-2} & \textbf{1.1e-2}  \\
 \hline
  Highschool (20:80) & 2.6e-4 &  \textbf{9.8e-3}  & 3.9e-3 & 7.6e-3 & 7.4e-3 \\
  \hline
 Primary School (80:20) & 3.2e-4 &  8.4e-3  &  3.3e-3 & 1.6e-2 & \textbf{1.7e-2}  \\ 
 \hline
  Primary School (60:40) & 1.1e-3 &  1.2e-2  & 9.9e-3 & \textbf{2.2e-2} & \textbf{2.2e-2} \\
 \hline
  Primary School (20:80) &  1.6e-3 & \textbf{3.0e-2}   & 1.6e-2 & 2.2e-2 & 2.1e-2 \\
   \hline
  Linear Model (80:20) & 6.4e-3 & \textbf{1.6e-1}  & 4.1e-2 & 7.8e-2 & 7.7e-2 \\ 
 \hline
  Linear Model (60:40) &  1.3e-2 & \textbf{2.8e-1}  &  8.5e-2 & 1.8e-1 & 1.8e-1  \\ 
 \hline
  Linear Model (20:80) &  3.2e-2 & \textbf{ 3.0e-1}  &  1.1e-1 & 1.4e-1 & 1.3e-1  \\ 

 \hline 

\end{tabular}
\end{center}
\caption{AUC-PR for triangle prediction on high-school contact data and synthetic hypergraphs from the linear model. Highest values are indicated in bold.}
\label{tbl:triangle_prediction}
\end{table}

\section{Conclusion}
\label{sec:con}
In this work we have developed new random models and embedding algorithms for 
hypergraphs, and investigated their equivalence. In particular, we focused on two spectral embedding algorithms customized for hypergraphs,
 which aim to reveal linear and periodic  structures, respectively. We also described random hypergraph models associated with these algorithms, which allow us to quantify the  relative strength of linear and periodic structures based on maximum likelihood. We demonstrated the model comparison approach on synthetic linear and periodic hypergraphs, showing that the results are consistent with the generating mechanism.
 When applied to high school and primary school contact hypergraphs, the model comparison suggests the periodic structure is more prominent. On this data set we also showed that the ``spectral embedding plus random hypergraph'' approach gives a useful strategy for predicting new hyperedges.
 
 In future work, it would be interesting to investigate how these linear and periodic hypergraph models compare with other versions that use 
 alternative assumptions, including those based on core-periphery \cite{tudisco2022core} and stochastic block model \cite{peixoto2022ordered,Chodrow21} structures.

\section*{Acknowledgements}
X.G. acknowledges support of MAC-MIGS CDT under EPSRC grant EP/S023291/1. D.J.H. was supported by EPSRC grant EP/P020720/1.
 K.C.Z. was supported by the Leverhulme Trust grant  2020-310 and EPSRC grant EP/V006177/1. 

\section*{Data, code and materials}
This research made use of public domain data that is available over the internet, as indicated in the text. Code for the experiments is available at \url{https://github.com/OpalGX/hypergraph_model}.

\section*{Competing interests}
The authors declare that they have no conflicts of interest. 

\section*{Authors' contributions}
X.G. carried out the numerical experiments and drafted the manuscript. All authors contributed to the theoretical understanding, development of research, and the writing of the manuscript. All authors gave final approval for publication.

\bibliographystyle{vancouver}
\bibliography{references}

\end{document}